\documentclass[preprint,12pt]{elsarticle}

\usepackage{graphicx}
\usepackage{amsfonts}

\usepackage{amssymb}	  

\usepackage{subfigure}

\usepackage{tikz,pgfplots}
\usepackage{subfigure}
\usepackage{array}
\usetikzlibrary{plotmarks}
\pgfplotsset{compat=newest}

\newcommand{\avg}[1]{\left< #1 \right>}              

\def\th{\theta}
\def\bth{\overline{\theta}}

\newtheorem{theorem}{Theorem}[section]

\newenvironment{proof}[1][Proof]{\begin{trivlist}
\item[\hskip \labelsep {\bfseries #1}]}{\end{trivlist}}

\newenvironment{remark}[1][Remark]{\begin{trivlist}
\item[\hskip \labelsep {\bfseries #1}]}{\end{trivlist}}

\journal{Physica D}

\begin{document}

\bibliographystyle{unsrt}  

\begin{frontmatter}

\title{Averaging and spectral properties for the 2D advection-diffusion equation in the semi-classical limit for vanishing diffusivity}

\author[GC]{E. Dedits}
\author[GC,CSI]{A. C. Poje}
\author[GC,CSI]{T. Sch{\"a}fer}
\author[GC,CSI]{J. Vukadinovic}

\address[GC]{Physics Program at the CUNY Graduate Center \\
                     365 5th Ave, New York, NY 10016}
\address[CSI]{Department of Mathematics, College of Staten Island \\
                      2800 Victory Blvd, Staten Island, NY 10314}

\begin{abstract}
We consider the two-dimensional advection-diffusion equation on a bounded domain subject to either 
Dirichlet or von Neumann boundary conditions and study both time-independent and time-periodic cases 
involving Liouville integrable Hamiltonians that satisfy conditions conducive to applying the averaging principle.    
Transformation to action-angle coordinates permits averaging in time and angle, leading to an
 underlying eigenvalue equation that allows for separation of the angle and action coordinates. 
The result is a one-dimensional second-order equation involving an anti-symmetric imaginary potential.  
For radial flows on a disk or an annulus, we rigorously  apply existing complex-plane WKBJ methods to study 
the spectral properties in the semi-classical limit for vanishing diffusivity.  
In this limit, the spectrum is found to be a complicated set consisting of lines related to Stokes graphs. 
Eigenvalues in the neighborhood of these graphs exhibit nonlinear scaling with respect to diffusivity  leading to
convection-enhanced rates of dissipation (relaxation, mixing) 
for initial data which are mean-free in the angle coordinate.
These branches coexist with a diffusive branch of eigenvalues that scale linearly with  diffusivity and contain the principal eigenvalue (no dissipation enhancement).   

\end{abstract}


\end{frontmatter}

\section{Introduction}

The focus of this paper is the Cauchy problem for the 2D advection-diffusion equation (ADE), a.k.a. passive scalar equation, which in the `small-diffusion' formulation reads 
\begin{equation}
  \label{eq:initial_ad_eqn}
  c_t -\varepsilon \Delta c+ ({\bf u}\cdot\nabla)c = 0, \ \ c(0,x,y)=c_0(x,y). 
\end{equation}
The unknown $c(t,x,y)$ is a scalar function of spatial coordinates $(x,y)$ and time $t$, while ${\bf u}(t,x,y)$ is a given time-independent or time-periodic vector-field (flow), and $\varepsilon>0$ is a given diffusivity.  We restrict our attention  to divergence-free flows only, i.e., the flows which satisfy $\nabla \cdot {\bf u}(t,\cdot)=0$.  More specifically, we assume the existence of a time-independent or time-periodic stream-function (Hamiltonian), $\Psi(t,x,y)$, such that ${\bf u}(t,x,y) = \nabla\times \Psi(t,x,y)$, with $\nabla\times = (\partial_y, -\partial_x)$  denoting the two-dimensional curl.  

The equation constitutes an important paradigm for a wide range of physical, chemical, and biological processes that are characterized both by transport induced by a fluid flow as well as diffusive forces of different nature.  Examples include homogenization in fluid mixtures, pollutant dispersion in the ocean or atmosphere, temporal evolution of biological systems in flowing media, energy transport in flowing media, etc.  Thus, it is at all not surprising that understanding of the advective dynamics of passive
scalars in the presence of diffusion has been a subject of intensive
research reaching back to at least as far as Batchelor \cite{batchelor}. 

A significant amount of physical and mathematical literature has been devoted to the study of the ADE in various settings: on unbounded domains, compact manifolds, or bounded domains in conjunction with certain boundary conditions.  The diffusion is responsible for a rather simple long-time dynamics characterized by the relaxation toward an 
equilibrium $c_{\rm eq}$,
\[
||c(t,\cdot)- c_{\rm eq}||_\infty\to0\ \ \ {\rm as}\  t\to\infty.
\]
($c_{\rm eq}\equiv0$ in the case of Dirichlet boundary conditions, or $c_{\rm eq}=\avg{c_0}$ in the case of von Neumann or periodic boundary conditions, for example.)
 
A large open area of study remains the rather intricate interplay between diffusion and advection in the semi-classical limit $\varepsilon\to0$. The convection-driven enhancement of dissipation (relaxation, mixing) is of particular interest.  Depending on the boundary problem, and closely related to it, on the structure of the spectrum of the non-self-adjoint advection-diffusion operator ${\mathcal L}_\varepsilon=- \varepsilon \Delta +{\bf u}\cdot\nabla $, various methods (e.g. homogenization, probabilistic, variational and `PDE' methods) have been employed.  Here, we review some of the relevant existing literature.   

On unbounded domains and periodic media, some of these questions have been addressed within the framework of the 
homogenization theory \cite{homogenization_thesis,homogenization1,homogenization2}.  It was shown that the long-time, long-distance behavior of solutions is governed by an effective-diffusion equation, $c_t=\Sigma_{i,j} a_{ij}^\varepsilon \frac{\partial^2c}{\partial x_i\partial x_j}$, where the 
constant matrix $A^\varepsilon=(a_{ij}^\varepsilon)$ is the so-called effective diffusivity tensor.  
The effective diffusion in a given direction ${\bf e}$ is then given by ${\bf e}^TA^\varepsilon{\bf e}$, and convection-enhanced scaling regimes as `good' as  $\varepsilon^{-1}$ and as `bad' as $\varepsilon$ have been identified.  In particular, the typical scaling regime $\varepsilon^{1/2}$ has been observed and referred to as convection-enhanced diffusion (see \cite{homogenization2}).  

Homogenization theory does not, however, provide satisfactory explanation of the short-term enhancement of 
dissipation by convection.  
In a recent paper \cite{zlatos}, Zlato\v s defined dissipation enhancing flows ${\bf u}$ on an unbounded 
domain $D$ by requiring  that for any initial condition $c_0\in L^p(D)$, the  solution $c^{(\varepsilon)}$ of the Cauchy problem (\ref{eq:initial_ad_eqn}) satisfy 
\begin{equation}\label{diss_enhance}
\lim_{\varepsilon\to0} ||c^{(\varepsilon)}(\varepsilon^{-1},\cdot)||_\infty=0.  
\end{equation}
Loosely speaking, Zlato\v s then characterized these flows  by the condition that the only eigenfunctions of the advection operator ${\bf u}\cdot\nabla$ are the first integrals of ${\bf u}$.  
 
For the purposes of this paper, we are primarily interested in the case of a bounded domain 
$D$ with Dirichlet (or von Neumann) boundary conditions. 
Under some mild regularity conditions on ${\bf u}$, the time-independent advection-diffusion operator possesses 
a pure point spectrum consisting of isolated eigenvalues that have positive (nonnegative) real part 
\cite{agmon, faierman}.  
The eigenvalue with the least positive real part is referred to as the principal eigenvalue, $\lambda_0^\varepsilon$. 
This value determines the slowest time scale of relaxation toward the equilibrium in the sense that, 
for generic initial data $c_0\in L^2(D)$ (i.e., the initial data with a non-vanishing projection onto the eigenspace corresponding to the principal eigenvalue), the following holds:  
\[
t^{-1}\log||c^{(\varepsilon)}(t,\cdot)||_{L^2(D)}\to -\lambda_0^\varepsilon\ \ \ {\rm as}\ \ \  t\to\infty. 
\]
Therefore, the dependence of $\lambda_0^\varepsilon$ on $\varepsilon$ essentially determines 
any convection-enhanced dissipation  (relaxation, mixing) rate for generic initial conditions.  
Berestycki et al.  \cite{berestycki} 
identified a sharp criterion for the principal eigenvalue $\lambda_0^\varepsilon$ to scale linearly with $\varepsilon$ as  $\varepsilon\to0$  (no dissipation enhancement).  More precisely, the authors proved that $\varepsilon^{-1}\lambda_0^\varepsilon$ is bounded as $\varepsilon\to0$ if and only if ${\bf u}$ has a first integral $w$ in the space $H_0^1(D)$, i.e., if ${\bf u}\cdot\nabla w=0$.  When this is not the case, the authors proved that we find ourselves in the dissipation enhanced regime in the sense that for any initial datum $c_0\in L^2(D)$, (\ref{diss_enhance}) holds. 

While the principal eigenvalue is the most relevant quantity, the complete picture of the relaxation dynamics 
is given by the  structure of the entire spectrum, potentially a very complicated set since the 
advection-diffusion operator ${\mathcal L}_\varepsilon$ is a sum of a self-adjoint operator, 
$-\varepsilon \Delta$, and an anti-self-adjoint operator, ${\bf u}\cdot\nabla$, and as such, 
possesses neither symmetry.  
The theory of non-self-adjoint (NSA) operators lags far behind the theory of self-adjoint (SA) operators (see \cite{davies:2007}).  Self-adjoint theory has at its disposal a powerful tool in the spectral theorem, as well as a variety of variational methods, that can be used to obtain tight bounds on eigenvalues both theoretically and 
numerically.  The self-adjoint theory and its techniques have been used to great effect in quantum mechanics.  The non-self-adjoint theory, on the other hand, is much less cohesive and still remarkably incomplete.  
It comprises a  wide variety of diverse methods whose only commonality is the use, in one way or another,  
of ideas from analytic function theory.

The complicated structure of the spectrum for ADE has been observed numerically by Giona et al. 
(see \cite{giona1, giona2}) where the most important feature is the occurrence  of different eigenvalue 
branches with possibly different convection-enhanced scaling regimes.  
In particular, for the parallel sine flow ${\bf u}(x,y)=(0,\sin(2\pi x))$ on the unit torus, a `diffusive' branch of eigenvalues which scale linearly with $\varepsilon$ is found to coexist with two (equivalent) `convective' branches of eigenvalues which scale as $\varepsilon^{1/2}$ -- the same scaling also observed by homogenization techniques (see \cite{homogenization2}) and referred to as convection-enhanced diffusion.  
Giona et al. provide a heuristic argument for the presence of the convection-enhanced branch based on the
numerically observed localization of the eigenfunctions. 
The aim of this paper is to provide a rigorous mathematical justification.  

In accordance with our discussion above about NSA operators, general results are extremely difficult to obtain.  
We focus on the case of Liouville integrable Hamiltonians $H(x,y)$ that allow for canonical transformation to 
action-angle coordinates $(J,\theta)$, and consider domains $D$ topologically equivalent to either a disk or an 
annulus whose boundary consists of level sets of $H(x,y)$.  
In this case, ${\bf u}=\nabla\times H$ possesses a first integral in $H_0^1(D)$ and hence the principal eigenvalue scales linearly with $\varepsilon$.  
The advantage of action-angle coordinate formulation is that the advection assumes a very simple form 
$\omega(J)c_\theta$ and, under mild conditions on $\omega(J)$,  solutions converge to solutions of the equation 
obtained by averaging the coefficients of the diffusion operator written in action-angle coordinates  with respect to the angle coordinate $\th$.  The underlying eigenvalue equation allows for separation of action and angle coordinates through the ansatz $c_{m}(J,\th)=e^{im\th}g_{m}(J)$, leading to a countable family of 1D equations in $g_{m}$ involving an imaginary potential $im\omega(J)$ ($i$ is the imaginary unit and $m$ is an integer).  For $m=0$, the equation is a self-adjoint problem leading to eigenvalues which scale linearly with $\varepsilon$.  When $m\not=0$, however, it is a non-self-adjoint convection-dominated problem which is much more difficult and leads to nonlinear scaling with respect to $\varepsilon$.  
In particular, the solutions for any initial data $c_0(J, \theta)$ with zero-mean in the angle coordinate $\th$, i.e., when $\int_0^{2\pi} c_0(J,\th)\ d\th\equiv 0$,  are subject to the dissipation enhanced regime in the sense of  (\ref{diss_enhance}) . Fixing $m=1$, the equation can be thought of as the 1D analogue of the advection-diffusion equation which retains exact features of the ADE not based on approximations. 
This equation also involves a non-self-adjoint operator whose spectral structure is difficult to
characterize in general.

An important example is  the case of unidirectional axisymmetric radial flows in annular regions or disks,  
which include physically realizable flows such as the Couette flow and Poiseuille flow. 
We are able to apply the WKBJ method on the complex plane, as developed in \cite{fedoryuk, shkalikov1, shkalikov2, shkalikov3}.   
The complexity of the spectrum depends on the behavior of the potential near the boundary and its critical points.  Following \cite{shkalikov1, shkalikov2, shkalikov3}, we consider the case when $\omega(r)$ is either monotone or concave, which includes the cases of (regularized) vortical flow, the Couette flow and the Poiseuille flow.  
We show that in the semi-classical limit $\varepsilon\to 0$, the spectrum converges to a collection of lines 
related to the so-called Stokes graphs.  We also infer information about the scaling of eigenvalues on those lines with respect to $\varepsilon$, which is typically either $\varepsilon^{1/2}$ (for quadratic critical points) 
or $\varepsilon^{1/3}$ (for non-critical boundary points).  
Note that the case of MHD driven annular micromixers leading to the Poiseuille profile was examined in \cite{gleeson1, gleeson2}.  Using a particular ansatz, the authors identified one spectral branch with the scaling $\varepsilon^{1/3}$.  As we shall see, scaling regimes $\varepsilon^{1/3}$ and $\varepsilon^{1/2}$ coexist for this particular case.  
 
In addition to treating the autonomous case, we also treat the time-periodic case, assuming, however, that the time dependence of the stream function $\Psi$ is completely separable so that
\begin{equation}
  \label{eq:stream_functio_separated}
  \Psi(t,x,y) = H(x,y)f(t).  
\end{equation}
The case where $f(t)$ is a mean-free periodic function of time with period $T$, 
\begin{equation}
  \label{eq:f(t)_mean_free}
  \avg{f} = \frac{1}{T}\int_0^T f(t) = 0, 
\end{equation}
was treated in \cite{ts_averaging}.  Lie-transformation based averaging techniques were used to show that the solutions converge to solutions of a autonomous self-adjoint diffusive equation. In this paper, we treat the case of non-vanishing mean $\avg{f}$.  We use the averaging principle for  perturbations of Hamiltonian systems (see \cite{freidlin1, freidlin2}) to show that the solutions converge to solutions of an autonomous advection-diffusion equation, so that the relaxation dynamics in this particular non-autonomous case are, up to a small diffusive correction,
 the same as those of the autonomous diffusion equation.

\section{Action-angle coordinates and averaging.}
\label{sec:stre-lines-coord-and-aver}
\subsection{Action-angle coordinates, the autonomous case, and  angle-averaging} 
Let us first consider the autonomous advection-diffusion equation (\ref{eq:initial_ad_eqn}) with ${\bf u}=\nabla\times H(x,y)$ for some time-independent Hamiltonian $H(x,y)$.  Throughout the paper, we assume that the level sets $M_h=\{(x,y):H(x,y)=h\}$ are compact, closed and connected curves, so that the conditions of the Liouville's theorem on integrable systems are satisfied.  It is well known (see \cite{arnold}) that the system then allows for a  canonical transformation to action-angle variables 
\begin{equation}
  \label{eq:canonical_transformation}
  \mathcal{C}: (x,y)\rightarrow(J, \theta), 
\end{equation}
which satisfy the following two conditions: (a) the Hamiltonian is a function of the action coordinate, $H=h(J)$ and (b) $\oint_{M_h}d\theta=2\pi$.  Let us first assume that $D$ is topologically equivalent to an annulus bounded by curves $\{(x,y):J(x,y)=j_1\}$ and $\{(x,y):J(x,y)=j_2\}$ with $j_1<j_2$.  We consider the Dirichlet boundary conditions $c|_{\partial D}=0$ (or the von Neuman boundary conditions $dc/dJ|_{\partial D}=0$). It is easy to see that ${\bf u}$ possesses a first integral belonging to $H_0^1(D)$, and therefore it is not dissipation enhancing in the sense of  (\ref{diss_enhance}) .  

Introducing $\omega(J)=h'(J)$, the original advection-diffusion equation (\ref{eq:initial_ad_eqn}) in action-angle 
coordinates can be written as 
\begin{equation}
  \label{eq:ad_in_AA_coordinates}
  c_t + \omega(J)c_{\theta} = \varepsilon\left(|\nabla_{(x,y)} J|^2c_{JJ}+|\nabla_{(x,y)} \theta|^2c_{\theta\theta}+(\Delta_{(x,y)} J) c_J+(\Delta_{(x,y)}\theta) c_\theta\right).
\end{equation}
Expressing the new  coefficients for the Laplacian in terms of action-angle coordinates,  $a_{11}(J,\theta)=|\nabla_{(x,y)} J|^2$, $a_{22}(J,\theta)=|\nabla_{(x,y)} \theta|^2$, $b_1(J,\theta)=\Delta_{(x,y)} J$ and $b_2(J,\theta)=\Delta_{(x,y)} \theta$, we write the equation as 
\begin{equation}
  \label{eq:ad_in_AA_coordinates}
  c_t + \omega(J)c_{\theta} = \varepsilon(A(J,\theta):\nabla\nabla + {\bf b}(J,\theta)\cdot\nabla)c.
\end{equation}
Here
\begin{eqnarray}
  \label{eq:inner_product_gamma}
  &A(J,\theta):\nabla\nabla = a_{11}(J,\theta)\partial_{JJ}  +a_{22}(J,\theta)\partial_{\theta\theta}, \\
  &{\bf b}(J,\theta)\cdot\nabla = b_1(J,\theta)\partial_{J}+b_2(J,\theta)\partial_\theta.
\end{eqnarray}
Note that the coefficient functions $a_{11}$, $a_{22}$, $b_1$ and $b_2$ are periodic in the angle coordinate $\theta$ with period $2\pi$.  We introduce the following notation for the averages 
\[
\overline a_{ij}(J)=\avg{a_{ij}}=\frac{1}{2\pi}\int_0^{2\pi}a_{ij}(J,\theta)\ d\theta
\]
and 
\[
\overline b_i(J)=\avg{b_{i}}=\frac{1}{2\pi}\int_0^{2\pi}b_{i}(J,\theta)\ d\theta.  
\] 
Note that $\overline b_2=0$.  We now write the averaged equation, 
\begin{equation}
  \label{eq:ad_av_in_AA_coordinates}
  c_t + \omega(J)c_{\theta} = \varepsilon(\overline A(J):\nabla\nabla +\overline{\bf b}(J)\cdot\nabla)c.
\end{equation}
If $d\omega/dJ$ has a finite number of zeros, then the solutions of (\ref{eq:ad_in_AA_coordinates}) converge to the solutions of (\ref{eq:ad_av_in_AA_coordinates}) in the limit $\varepsilon\to0$ (see \cite{freidlin1,freidlin2, freidlin3, freidlin4, koralov}).  We are interested in the underlying eigenvalue problem 
\begin{equation}
  \label{eq:ad_av_in_AA_coordinates_eigenvalue}
- \varepsilon(\overline A(J):\nabla\nabla +\overline{\bf b}(J)\cdot\nabla)c+ \omega(J)c_{\theta} =\lambda c.
\end{equation}
This equation has the advantage that we can seek the eigenfunctions through the ansatz  
\[
c_{m,n}(J,\th)=e^{im\th}g_{m,n}(J), 
\]
where $m$ and $n$ are integers and $g_{m,n}(J)$ satisfies the eigenvalue problem 
\begin{equation}
  \label{eq:ad_av_in_AA_coordinates_eigenvalue_1D}
- \varepsilon\left(\overline a_{11}(J)g''_{m,n} +\overline b_{1}(J)g'_{m,n}-m^2\overline a_{22}(J)g_{m,n}\right)+ im\omega(J)g_{m,n}  =\lambda g_{m,n}, 
\end{equation}
together with the Dirichlet boundary condition $g_{m,n}(j_1)=g_{m,n}(j_2)=0$ (or von Neumann boundary conditions $g'_{m,n}(j_1)=g'_{m,n}(j_2)=0$).   
For $m=0$, this equation reads 
\begin{equation}
  \label{eq:ad_av_in_AA_coordinates_eigenvalue_1D}
- \varepsilon\left(\overline a_{11}(J)g''_{0,n} +\overline b_{1}(J)g'_{0,n}\right) =\lambda g_{0,n},
\end{equation}
which is a self-adjoint problem leading to eigenvalues which scale linearly with $\varepsilon$.  When $m\not=0$, however, it is a non-self-adjoint convection-dominated problem which is much more difficult to analyze and leads to nonlinear scaling with respect to $\varepsilon$.  As already discussed in the introduction, the solutions for any initial data $c_0(J, \theta)$ with zero-mean in the angle coordinate $\th$, i.e., when $\int_0^{2\pi} c_0(J,\th)\ d\th=0$,  are subject to the dissipation enhanced regime in the sense of  (\ref{diss_enhance}) .

The case when $D$ is topologically equivalent to a disk, i.e., when the boundary $\partial D=\{(x,y):J(x,y)=j_2\}$ for some $j_2$ is treated similarly.  Then, there exists $j_1$ (w.l.o.g., we assume that $j_1<j_2$), so that $\{(x,y):J(x,y)=j_1\}$ consists of a point at which $\theta(x,y)$ is undefined.  We then study the problem on a punctured disk $\{(x,y):j_1<J(x,y)<j_2\}$.  The difficulty, however, arises from the fact that the coefficients of the diffusion operator in action-angle coordinates possess a singularity at the punctured point, and one has to infer the boundary condition at $j_1$ from the asymptotic behavior of the solution at the punctured point.  

\subsection{The time-periodic case and  time-averaging.} 
Let $f(t)$ be $T$-periodic for some $T>0$, and let ${\bf u}(t,x,y)=\nabla\times \Psi(t,x,y)$, where 
\[
  \Psi(t,x,y) = H(x,y)f(t).  
\]
This paper is a continuation of authors' work \cite{ts_averaging} for the mean-free case $\avg{f}=\int_0^T f(t)\ dt=0$, and in this section we review some of the results of that paper.   The functon $f(t)=f_0(t) + f_1$ is assumed to be time-periodic with period $T>0$.  We assume that $f_0(t)$ is periodic and mean-free function of time $t$ and $f_1$ is a constant. There are three dynamically distinct cases: (1) the autonomous case ($f_0\equiv0$); (2) the non-autonomous case with vanishing mean ($f_1\equiv0$); and (3) the non-autonomous case with non-vanishing mean ($f_0\not\equiv0$ and $f_1\not\equiv0$). Firstly, we introduce 
\begin{eqnarray}
  \label{eq:func_r}
  F(t) = \int_0^t f_0(t')dt', 
\end{eqnarray}
which is also $T$-periodic.  Again, we transform the equation using action-angle coordinates, 
\[
  \mathcal{C}: (x,y)\rightarrow(J, \theta).
\]
Similarly as before, the original advection-diffusion equation (\ref{eq:initial_ad_eqn}) in these 
coordinates is written as
\begin{equation}
  \label{eq:ad_in_AA_coordinates1}
  c_t + f(t)\omega(J)c_{\theta} = \varepsilon(A:\nabla\nabla + {\bf b}\cdot\nabla)c.
\end{equation}
We now use stream-lines $\overline{J} = J$ and  $\bth = \theta -
\omega(J)F(t)$ as new coordinates via the transformation 
\begin{equation}
  \label{eq:change_of_coordinates_to_streamlines}
  c(t,J,\th) \longrightarrow v(t,\overline{J},\bth)
\end{equation}
with the following transformation rules:
\begin{eqnarray*}
  \label{eq:transformation-rules}
 & c_t = v_t - v_{\bth}\omega f,\;\;\;\;
  c_{\th} = v_{\bth}, \;\;\;\;c_{\th\th} = v_{\bth\bth},\\
 &  c_J = v_J - v_{\bth}\omega' F, \;\;\;\;
  c_{JJ} = v_{JJ} - v_{\bth}\omega'' F - 2v_{J\bth}\omega' F
  + v_{\bth\bth}(\omega' F)^2. 
\end{eqnarray*}
This transformation to stream-lines coordinates is nothing but a
transformation to a new ``co-moving'' reference frame. Denoting by $\tilde A$ and $\tilde{\bf b}$ the (time-dependent) coefficient matrix and vector in these new coordinates, the equation (\ref{eq:ad_in_AA_coordinates1})  becomes
\begin{equation}
  \label{eq:new_equation_in_streamlines_coords1}
  v_{t} + f_1\omega(J)v_{\bth} =(\tilde A:\nabla\nabla + \tilde{\bf b}\cdot\nabla)v. 
\end{equation}
When $f_1=0$,  the advective term in equation (\ref{eq:ad_in_AA_coordinates1})
disappears, and we  obtain an equation for $v$ of the form
\begin{equation}
  \label{eq:new_equation_in_streamlines_coords}
  v_{t} =\varepsilon(\tilde A:\nabla\nabla + \tilde{\bf b}\cdot\nabla)v. 
\end{equation}
All effects of the influence of the advective field are now contained in the time-dependent
coefficients $\tilde{A}$ and $\tilde{\bf b}$ and, therefore, equation 
(\ref{eq:new_equation_in_streamlines_coords}) is now suitable for
averaging.     The main idea of the authors' paper \cite{ts_averaging} was to apply a near-identity Lie transform that
eliminates the explicit time dependence of the coefficients (see
\cite{ts_averaging} for details). However, in the case when $f_1\not=0$, we write the averaged equation in an ad-hoc fashion,   
\begin{equation}
  \label{eq:averaged_first_order1}
  v_{\tau} +f_1\omega(J)v_{\bth}= \varepsilon\Bigl(\avg{\tilde{A}}:\nabla\nabla
  +\avg{\tilde{\bf b}}\cdot\nabla \Bigr) v.
\end{equation}
The justification for convergence to the averaged equation via the Lie-transform approach, which was used for the zero-mean case does not apply to this case.  The convergence to the averaged equation on timescales $\varepsilon^{-1}$, however, can be justified by the averaging principle, given that the Hamiltonian $H(x,y)$ satisfies some mild regularity conditions (see \cite{freidlin1}, Theorem 3.2). 
\subsection{Radial flows.}
\label{sec:vortical_flow}
As already mentioned in the introduction, we are particularly interested in a special case that is very instructive both analytically and numerically, the case of unidirectional axisymmetric radial flows in annular regions or disks,  which include physically realizable flows such as the Couette and Poiseuille flows, and the (regularized) vortical flow.  We  assume that the time dependance is completely separable; in particular let us assume that the stream function is given by
\begin{equation}
  \label{eq:stream_function}
  \Psi(t,x,y) =  H(x,y)f(t) = h(r)f(t), 
\end{equation}
where $r=\sqrt{x^2+y^2}$ is the radial coordinate. For uniaxial radial flows, the action-angle variables can be expressed via the usual polar coordinates, $(x,y) \rightarrow (r^2/2,\theta)$.  Denoting by $\nu(r)=h'(r)$ the azimuthal velocity and by $ \omega(r) = \nu(r)/r$ the ``potential'', the advection-diffusion equation (\ref{eq:initial_ad_eqn}) then reads in polar coordinates
\begin{equation}
c_t + f(t)\omega(r)c_{\theta} -\varepsilon\Delta\,c=0\,,
\end{equation}
where  $\Delta\,c=\left(\frac{1}{r}c_r + c_{rr}+\frac{1}{r^2}c_{\theta\theta}\right)$ is the Laplace operator in polar coordinates.  In the non-autonomus case  $f_0\not\equiv0$, we introduce  
\begin{eqnarray}
  \label{eq:func_r}
  F(t) = \int_0^t f_0(t')dt'\,,
\end{eqnarray}
and we derive the stream-lines equations
\begin{equation}
  \label{eq:streamLines_equations}
  dx/dF  = \omega(r) y, \qquad dy/dF = -\omega(r) x\,.
\end{equation}
We can now use the stream-lines $\overline{r} = r$ and  $\bth = \theta -
\omega(r)F(t)$ as the new coordinates via the transformation 
\begin{equation}
  \label{eq:change_of_coordinates_to_streamlines}
  c(t,r,\th) \longrightarrow v(t,\overline{r},\bth).
\end{equation}
We obtain
(\ref{eq:new_equation_in_streamlines_coords1}) in form
\begin{equation}
  \label{eq:AD_ready_for_averaging}
  v_t - \varepsilon\Bigg(\Delta v + F\Bigg(\Bigl(\frac{\omega'}{r} + \omega''\Bigr)v_{\bth}
  + 2\omega'v_{\bth r}\Bigg) + F^2(\omega')^2v_{\bth\bth}\Bigg)+ f_1\omega v_{\bth} = 0. 
\end{equation}
We again write the averaged
counterpart of (\ref{eq:AD_ready_for_averaging})  by simply replacing the time-dependent coefficients by their time averages as
\begin{equation}
  \label{eq:AD_averaged}
  V_{t} - \varepsilon\Bigg(\Delta V + \avg{F}\Bigg(\Bigl(\frac{\omega'}{r} + \omega''\Bigr)V_{\bth}
  + 2\omega'V_{\bth r}\Bigg) + \avg{F^2}(\omega')^2V_{\bth\bth}\Bigg)+f_1\omega V_{\bth} = 0. 
\end{equation}
If instead of the above transformation we use $\bth = \theta +
\omega(r)(F(t)-\avg{F})$, the averaged equation assumes the simpler form 
\begin{equation}
  \label{eq:AD_averaged1}
  V_{t} - \varepsilon\Bigg(\Delta V+ \avg{(F-\avg{F})^2}(\omega')^2V_{\bth\bth}\Bigg)+f_1\omega V_{\bth} = 0\,.
\end{equation}
Note that this averaged equation is essentially the autonomous advection-diffusion equation up to a `small' diffusive correction in the $\bth$ variable, whose contribution does not change the spectral scaling properties.    

\section{Spectral properties of the autonomous operator} 
\label{sec:sperctral}
In the following section, we study the spectral properties of the autonomous non-selfadjoint operator associated with the equation (\ref{eq:AD_averaged1}), i.e. the problem, 
\[
-\varepsilon\Delta V(r,\th)+\omega(r) V_\th(r,\th)=\lambda V(r,\th)
\]
on a disk $D=\{0\le r \le r^+, 0\le \theta \le 2\pi\}$ or an annulus $D=\{0<r^-\le r \le r^+, 0\le \theta \le 2\pi\}$ 
subject to homogeneous Dirichlet boundary conditions $V|_{\partial D}=0$ or the von Neumann boundary conditions $\frac{\partial V}{\partial r}|_{\partial D}=0$. Here, $\Delta$ stands for the Laplacian in polar coordinates, 
\[
\Delta V=V_{rr}+\frac1r V_r+\frac1{r^2}V_{\th\th}. 
\]
This eigenvalue problem allows for the separation of polar coordinates, and we seek the eigenfunctions in the form  
\[
V(r,\th)=e^{im\th}g_{m,n}(r), 
\]
where $m$ and $n$ are integers and $g_{m,n}$ satisfies the one-dimensional eigenvalue problem 
\begin{equation}
\label{eq:eigenvalue}
L_{m,\varepsilon} g_{m,n}=\lambda_{m,n} g_{m,n},  
\end{equation}
with
\[
(L_{m,\varepsilon} g)(r):=-\varepsilon\left(g''(r)+\frac{1}{r}g'(r)-\frac{m^2}{r^2}g(r)\right)+im\omega(r)g(r). 
\]
In the case of the annulus, the boundary condition becomes  $g_{m,n}(r^{\pm})=0$  (Dirichlet) or $\frac{d}{dr}g_{m,n}(r^{\pm})=0$ (von Neumann). Note that in the case of the disk, the singularity at $r=0$ is regular, leading to the the asymptotic behavior $g_{m,n}(r)\sim {\rm const}\cdot r^m$ as $r\to0$.  Therefore, it is plausible to impose the boundary condition $\lim_{r\to0^+}  \frac{g_{m,n}(r)}{r^m}=1$. 

If $m=0$, (\ref{eq:eigenvalue}) is a self-adjoint problem, 
\[
-\varepsilon(r^2g_{0,n}''(r)+rg_{0,n}'(r))=\lambda_{0,n} r^2g_{0,n}(r).  
\]
Recall that in the case of the disk and the Dirichlet boundary conditions, the eigenvalue--eigenfunction pairs are $\lambda_{0,n}=\varepsilon (j_{0,n}/r^+)^2$, $g_{0,n}(r)=J_0(j_{0,n}r/r^+)$, where $J_0$ is the Bessel function and $\{j_{0,n}\}$ are its positive zeros in the increasing order.  In general, the eigenfunctions are sought-after in the form $g_{0,n}(r)=c_1 H_0^{(1)}(\sqrt{\lambda_{0,n}}r)+c_2H_0^{(2)}(\sqrt{\lambda_{0,n}}r)$, where $H_m^{(1)}$ and $H_m^{(2)}$ are Hankel functions of order $m$.  In either case, the eigenvalues scale according to $\lambda_{0,n}\sim {\rm const}\cdot\varepsilon$.  When $m\not=0$, the problem is no longer self-adjoint, but rather it involves a sum of a self-adjoint and a anti-self-adjoint operator.  


\subsection{WKBJ approximations and the semi-classical limit}
In order to investigate the eigenvalue problem in the semi-classical limit $\varepsilon\to0$, we need to rewrite (\ref{eq:eigenvalue}) in a form suitable for applying the so-called WKBJ method.  To this end, we introduce the change of variables $s=r^{-2m}$.  Letting $h(s)=\frac{g(r)}{r^m}$, one can  verify that
\[
g''(r)+\frac{1}{r}g'(r)-\frac{m^2}{r^2}g(r)=4m^2r^{-(3m+2)}\frac{d^2h(s)}{ds^2}
\]
Therefore, $\lambda_{m,n}$ and $g_{m,n}$ satisfy (\ref{eq:eigenvalue}) if and only if $\mu_{m,n}=-\frac{i}m\lambda_{m,n}$ and $h_{m,n}(s)=\frac{g_{m,n}(r)}{r^m}$ satisfy the equation 
\begin{equation}
  \label{eq:eigenvalue1}
  4mi\varepsilon s^{\frac{2m+1}{m}} \frac{d^2h_{m,n}}{ds^2}=\left(\mu_{m,n}-\tilde\omega(s)\right) h_{m,n},  
\end{equation}
where we introduced the notation $\tilde\omega(s)=\omega(r)$. 
The  annulus problem is now posed on the interval $[s^{-}=(r^+)^{-2m},s^+=(r^-)^{-2m}]$, and the boundary conditions become either $h_{m,n}(s^\pm)=0$ (Dirichlet) or $\frac{d}{ds}(s^{-1/2}h_{m,n}(s))\big|_{s=s^\pm}=0$ (von Neumann). For the disk problem, the equation is now posed on $[s^-,s^+=+\infty)$, and the asymptotic behavior at the singularity leads to the boundary condition $h_{m,n}(+\infty)=1$.  

We define  
\[
(\tilde L_{m,\varepsilon}h)(s):=4mi\varepsilon s^{\frac{2m+1}{m}}\frac{d^2h(s)}{ds^2}+\tilde\omega(s)h(s), 
\]
and let 
\[
\avg{h_1,h_2}:=\int_{s^-}^{s^+}s^{-\frac{2m+1}{m}}h_1(s)h_2(s)\ ds
\]
Let us first assume that $\omega$ is a strictly monotonic (w.l.o.g., increasing) function and let $[a^-,a^+]=\omega([r^-,r^+])$ be its range.  Let 
$\Pi=\{\mu\ \big|\ \Im(\mu)<0,\ \Re(\mu)\in(a^-,a^+)\}$ be a semi-strip in the complex plane. It is obvious that 
$
\{\langle\tilde L_{m,\varepsilon}h,h\rangle|\avg{h,h}=1\}\subset \Pi,
$
and hence the eigenvalues of the problem lie in the semi-strip $\Pi$, as well.  

Following \cite{shkalikov2}, we make additional assumptions on $\omega$: 
\begin{enumerate}
\item Let there be a domain $G\subset\mathbb{C}$
such that $\omega$ is analytic in $G$ and maps $\overline G$ bijectively onto $\overline\Pi$. W.l.o.g., we can assume that $G$ lies entirely below the real axis.  
\item
For any $c\in (a^-,a^+)$, the preimage under $\omega$ of the ray $\{\mu=c-it\big|t\geq0\}$ is a function with respect to the imaginary axis. 
\item $\omega$ is analytic on some $\sigma$-neighborhood $U_\sigma$ of the segment $[r^-,r^+]$. 
\end{enumerate}
Let $\tilde G$ be such that the function $\zeta\mapsto\zeta^{-1/(2m)}$ is a bijection between $\tilde G$ and $G$, and in the following, let us fix that branch. Note that $G$ and $\tilde G$ can be chosen so that the conditions 1. through 3. hold for both $\omega:G\to\mathbb{C}$ and $\tilde\omega:\tilde G\to\mathbb{C}$.  

We now define functions, which are used to construct the WKBJ approximations of solutions to (\ref{eq:eigenvalue1}) (see for example \cite{fedoryuk}).  For $\mu\in\overline\Pi$, let $r_\mu$ denote the turning point of $\omega(r)-\mu$, i.e., let it be (the only) root of the equation $\omega(r)-\mu=0$. We define  
\[
S(r,\mu)=\int_{r_\mu}^{r}\sqrt{i\left(\omega(\xi)-\mu\right)}\ d\xi.  
\]  
For a fixed $\mu$, $S(r,\mu)$ is a multi-valued function.   
It is analytic on $\Pi$ and continuous on $\overline \Pi$ with respect to the variable $\mu$ and  locally analytic with respect to $r\in G$ with the branch point $r_\mu$. For a fixed $\mu\in\Pi$, we define the Stokes lines outgoing from $r_\mu$ as the analytic curves of the level set $\Re S(r,\mu)=0$ initiating at the turning point $r_\mu$.  It can be shown (see \cite{shkalikov2}) that in our particular case there are three Stokes lines initiating out of the turning point $r_\mu$ -- the `left', $\ell_{\rm left}$, the `right', $\ell_{\rm right}$, and the `lower', $\ell_{\rm lower}$. The maximal connected component $\mathcal{C}_\mu=\ell_{\rm left}\cup\ell_{\rm right}\cup\ell_{\rm lower}$ of the level set $\Re S(r,\mu)=0$ that includes the point $r_\mu$ is referred to as the Stokes complex, while the entire level set is referred to as the Stokes graph.  For a fixed $\mu$, we say that a domain $\Omega_\mu$ is canonical if the function $r\to S(r,\mu)$ is univalent on $\Omega_\mu$.  It follows easily that domains that contain points from one Stokes line only are canonical. In this particular case, we can identify three maximal canonical domains,  each of which has one of the Stokes lines from the Stokes complex belonging to it, while the other two as on its boundary.  We denote each one of these domains by $\Omega_\mu^{\rm left}$,  $\Omega_\mu^{\rm right}$ or $\Omega_\mu^{\rm lower}$, depending on which one of the three Stokes lines belongs to it. The branch of the function $S(r,\mu)$ for which $\Im S(r,\mu)\geq0$ on the Stokes line belonging to the canonical domain is also said to be the canonical branch for that canonical domain. However, for practical purposes, we will deviate from this convention.  Note that $S(r,\mu)$ can be extended analytically on either but not simultaneously on both sides of the Stokes lines on its boundary.  

In a similar fashion, on $\overline\Pi\times \tilde G$, we define 
\[
\tilde S(s,\mu)=\frac{1}{2m}\int_{s_\mu}^s\sqrt{i\left(\tilde\omega(\zeta)-\mu\right)}\zeta^{-\frac{2m+1}{2m}}\ d\zeta,
\]  
where $\tilde\omega(s_\mu)=\mu$.  Note that with $s=r^{-2m}$,  $\tilde S(s,\mu)=S(r, \mu)$. From the general WKBJ theory applied to the equation (\ref{eq:eigenvalue1}) easily follows that it possesses two so-called WKBJ approximations of the form 
\[
h_{\rm{app}}^{\pm}(s,\mu)=\frac{s^{\frac{2m+1}{4m}}}{\sqrt[4]{{i}\left(\tilde\omega(s)-\mu\right)}}e^{\pm (m\varepsilon)^{-\frac12}\tilde S(s,\mu)}, 
\]
which lead to the WKBJ approximations of (\ref{eq:eigenvalue})
\begin{equation}\label{approximate_solutions}
g_{\rm{app}}^{\pm}(r,\mu)=\frac{r^{-1/2}}{\sqrt[4]{{i}\left(\omega(r)-\mu\right)}}e^{\pm (m\varepsilon)^{-\frac12}S(r,\mu)}.  
\end{equation}
In the following we assume that $m\ll\frac1\varepsilon$.  We will use the Birkhoff notation $[1]^\pm=1+O^{\pm}(\varepsilon^{\frac12})$. 
\begin{theorem}
Given $\mu\in\Pi$, Eq. (\ref{eq:eigenvalue}) possesses two linearly independent solutions of the form 
\begin{equation}\label{fundamental_solutions}
g^{\pm}(r,\mu)^\pm=g_{\rm{app}}^{\pm}(r,\mu)(1+O^{\pm}(\varepsilon^{\frac12})), 
\end{equation}
where $O^\pm$ satisfies $|O^{\pm}(\varepsilon^{\frac12})|\leq C\varepsilon^{\frac12}$, with a constant $C$ not depending on $r$ as it varies on a compact set $K$ belonging to a canonical domain $\Omega_\mu$.  Moreover, the constant $C$ does not depend on $\mu$ or $r$ as they vary on compact sets $K'\subset\overline\Pi$ and $K\subset\cap_{\mu\in K'}\Omega_\mu$, respectively.  
\end{theorem}
\begin{proof} See \cite{shkalikov2}. \qed
\end{proof}

We now introduce the functions 
\[
Q^\pm(\mu)= \pm\int_{r_\mu}^{r^\pm}\sqrt{i\left(\omega(\xi)-\mu\right)}\ d\xi\ \  {\rm and}\ \   Q^\infty(\mu)=Q^+(\mu)+Q^-(\mu).
\]
 We fix the branches by the condition that for $c\in (a^-,a^+)$, $Q^+(c)=e^{i\pi/4}\alpha_c^+$ with $\alpha_c^+>0$ and $Q^-(c)=e^{i\pi/4}\alpha_c^-$ with $\alpha_c^-<0$. We define the sets 
\[
\tilde\gamma_\pm = \{\mu\in\overline\Pi\ |\ \Re Q^\pm(\mu)=0\} \ \ \rm{and} \ \  \tilde\gamma_\infty = \{\mu\in\overline\Pi\ |\ \Re Q(\mu)=0\}. 
\]  
Note that the definition would suggest that $Q^\infty(\mu)=S(r^+,\mu)-S(r^-,\mu)$, however this is only true if the segment $[r^-,r^+]$ is contained within a canonical domain.  
In the next theorem, we summarize some important properties of these functions.  
\begin{theorem}\label{gamma}
The curves $\tilde\gamma_-$ and $\tilde\gamma_+$ pass through the points $a$ and $b$, respectively.  Both curves are one-to-one with respect to the interval $[a^-,a^+]$, while the curve $\tilde \gamma_{\infty}$ is one-to-one with respect to $[0,-i\infty)$ on the imaginary axis.  The functions $Q^+(\mu)$,  $Q^-(\mu)$ and $Q^\infty(\mu)$ are univalent in the semi-strip $\Pi$, and, consequently, $\Im Q^+(\mu)$,  $\Im Q^-(\mu)$ and $\Im Q^\infty(\mu)$ are strictly monotonic on $\tilde\gamma_+$, $\tilde\gamma_-$ and $\tilde\gamma_\infty$, respectively.   The function $\Re Q^+(\mu)$ ($\Re Q^-(\mu)$) is positive (negative) above the curve $\tilde\gamma_+$  ($\tilde\gamma_-$), and it is of the opposite sign bellow that curve.  

The three curves have a unique intersection point (knot) $\mu_0$.  We denote by $\gamma_+$, $\gamma_-$ and $\gamma_\infty$ the parts of  $\tilde\gamma_-$, $\tilde\gamma_+$ and $\tilde\gamma_\infty$ between the knot $\mu_0$ and the points $a$, $b$ and $-i\infty$ respectively.  Let $\tilde\Gamma=\tilde\gamma_-\cup\tilde\gamma_+\cup\tilde\gamma_\infty$ and  $\Gamma=\gamma_-\cup\gamma_+\cup\gamma_\infty$. For $k\in\mathbb{Z}$, let $\mu_k^+$, $\mu_k^-$ and $\mu_k^\infty$ denote the solutions of $\Im Q^+(\mu)=(m\varepsilon)^{\frac12}(k\pi-\pi/4)$, $\Im Q^-(\mu)=-(m\varepsilon)^{\frac12}(k\pi-\pi/4)$ and $\Im Q^\infty(\mu)=(m\varepsilon)^{\frac12}k\pi$. For these three equations, let $p_\pm$, $m_\pm$ and $s_0$ be the indices so that $\{\mu_k^+\}_{p_+}^{q_+}$, $\{\mu_k^-\}_{p_-}^{q_-}$, and $\{\mu_k^{\infty}\}_{p_\infty}^{\infty}$ are all the solutions belonging to $\gamma_+$, $\gamma_-$ and $\gamma_\infty$, respectively.  We shall abuse the notation somewhat and assume that $0\leq p_\pm\leq q_\pm$, so that $\Im \mu_{p_\pm}$ are maximal.  
\end{theorem}
The set $\Gamma$ is the limit spectral graph of Eq. (\ref{eq:eigenvalue}) with the Dirichlet boundary conditions, in the sense of the following theorem.  
\begin{theorem}
Given $\delta>0$ there exists $\varepsilon_0>0$ such that all the Dirichlet eigenvalues of Eq. (\ref{eq:eigenvalue}) lie in the $\delta$-neighborhood $\Gamma_\delta$ of $\Gamma$ provided that $0<\varepsilon<\varepsilon_0$.  
\end{theorem}
\begin{proof} 
Here we just outline a sketch of the proof.  For a given $\delta$, the set $\Pi\backslash\Gamma_\delta$ consists of three disjoint, connected, closed components, the `left', $\Lambda_l$, the `right', $\Lambda_r$ and the `upper', $\Lambda_u$.  We need to show that there is a number $\varepsilon_0>0$ so that for $0<\varepsilon<\varepsilon_0$, the characteristic determinant 
\[
\Delta(\mu)=\left| \begin{array}{cc} g^+(r^+,\mu) & g^+(r^-,\mu) \\
g^-(r^+,\mu) & g^-(r^-,\mu) \end{array} \right|
\]
 for the fundamental solutions (\ref{fundamental_solutions}) does not vanish on $\mu\in\Pi\backslash\Gamma_\delta$. It can be proven (see \cite{shkalikov2}) that for  $\mu\in\Pi\backslash (\gamma_-\cup\gamma_+)$ there exists a canonical domain $\Omega_\mu$ and a path $\gamma_\mu$ within $\Omega_\mu$ which connects the points $r^-$ and $r^+$.  Hence, the following representation holds:  
\begin{equation}
\Delta(\mu)=\frac{T(\mu)}{r}\left(e^{(m\varepsilon)^{-\frac12}(S(r^+,\mu)-S(r^-,\mu))}[1]-e^{-(m\varepsilon)^{-\frac12}(S(r^+,\mu)-S(r^-,\mu))}[1]\right)
\end{equation}
where $T(\mu)=(i(\omega(r^-)-\mu))^{-1/4}(i(\omega(r^+)-\mu))^{-1/4}$, which does not vanish in $\Pi$.  Therefore, $\Delta(\mu)=0$ is only possible if 
\[
e^{\pm(m\varepsilon)^{-\frac12}(S(r^+,\mu)-S(r^-,\mu))}=1+O(\varepsilon^{\frac12}). 
\]
A sufficient condition for $\Delta(\mu)\not=0$ is therefore $|\Re(S(r^+,\mu)-S(r^-,\mu))|>C(\delta)$ with a constant $C(\delta)$ depending only on $\delta$.  The proof requires a separate discussion depending on the position of $\mu$ with respect to the lines $\tilde\gamma_\pm$. In the case when $\mu$ is above  $(\tilde\gamma_-\backslash\gamma_-)\cup(\tilde\gamma_+\backslash\gamma_+)$, we use the fact that $r^-$ and $r^+$ are connected by $\gamma_\mu$ which  intercepts  one Stokes line only, at which $\Re S(r,\mu)$ changes sign.  Therefore, $\Re S(r^+,\mu)$ and $\Re S(r^-,\mu)$ are of opposite signs.  We use compactness arguments to arrive at the conclusion.  For $\mu$ under $(\tilde\gamma_-\backslash\gamma_-)\cup(\tilde\gamma_+\backslash\gamma_+)$, $\gamma_\mu$ can be chosen to be the segment $[r^-,r^+]$, and therefore w.l.o.g., $S(r^+,\mu)-S(r^-,\mu)=Q^\infty(\mu)$. $S(r^+,\mu)$ and $S(r^-,\mu)$ are of the same sign; however, we make use of the fact that $\Lambda_l\cup\Lambda_r$ is $\delta$-distance from the set where $\Re(Q^\infty(\mu))=0$ to arrive at the same conclusion.   Note that $\Delta(\mu)\not=0$ in this case is equivalent to  
\begin{equation}\label{Q}
e^{\pm(m\varepsilon)^{-\frac12}Q^\infty(\mu)}=1+O(\varepsilon^{\frac12}). 
\end{equation}
Further details can be found in \cite{shkalikov2}. \qed
\end{proof}
The last theorem can be strengthened in the sense that the Dirichlet eigenvalues of Eq. (\ref{eq:eigenvalue}) can be tracked by points $\mu_k^+$, $\mu_k^-$ and $\mu_k^\infty$, which lie on $\Gamma$ ($\gamma_+$, $\gamma_-$ and $\gamma_\infty$, respectively), and their scaling with $\varepsilon$ yields the scaling of the eigenvalues.  
\begin{theorem}
Let $\delta>0$ and let $\{\mu_k^+\}_{p_+}^{q_+}$, $\{\mu_k^-\}_{p_-}^{q_-}$, and $\{\mu_k^{\infty}\}_{p_\infty}^{\infty}$ be as in Theorem \ref{gamma}. Consider $\{\mu_k^+\}_{p_+'}^{q_+'}$, $\{\mu_k^-\}_{p_-'}^{q_-'}$, and $\{\mu_k^{\infty}\}_{p_\infty'}^{\infty}$ consisting of those points which lie outside $U_\delta(a)\cup U_\delta(b)\cup U_\delta(\mu_0)$.  Then there exists $C=C(\delta)>0$ such that every Dirichlet eigenvalue of Eq. (\ref{eq:eigenvalue}) lies either in $U_\delta(a)\cup U_\delta(b)\cup U_\delta(\mu_0)$ or in a $C\varepsilon$ neighborhood of one of the points from those three sets. Each neighborhood contains at most one eigenvalue.  
\end{theorem}
\begin{proof}
Again, we only sketch the proof.  Let us focus on the curve $\gamma_+$.  For the sake of clarity, let us assume that $\mu\in\gamma_+\backslash(U_\delta(b)\cup U_\delta(\mu_0))$. Recall that with each $\mu\in\Pi$ a three-line Stokes complex $\mathcal{C}_\mu=\ell_{\rm left}\cup\ell_{\rm right}\cup\ell_{\rm lower}$ is associated. These yield three distinct canonical domains, each of which  has one of the Stokes lines from the complex belonging to it, and the other two on its boundary.  We denote each one by $\Omega_\mu^{\rm left}$,  $\Omega_\mu^{\rm right}$ or $\Omega_\mu^{\rm lower}$, depending which one of the three Stokes lines belongs to it. Let us consider the two sets of fundamental solutions $g^\pm_{\rm left\slash right}$ associated with canonical domains $\Omega_\mu^{\rm left}$ and $\Omega_\mu^{\rm right}$.  For $\mu\in\Pi$, $r^+\in \Omega_\mu^{\rm right}$ and $r^-\in \Omega_\mu^{\rm left}$. Actually, based on our assumption $\mu\in\gamma_+$, we have $r^+\in\ell_{\rm right}$. In the expression (\ref{fundamental_solutions}) for $g^\pm_{\rm left}$, we fix a branch by the condition $S(r^-,\mu)=\alpha(\mu)<0$ for $\mu$ sufficient close to $\gamma_+$.  

The two fundamental sets of solutions (\ref{fundamental_solutions}) are related through the following transmission formula for neighboring canonical domains: 
\[
\left( \begin{array}{c} g^+_{\rm left}(r,\mu) \\
g^-_{\rm left}(r,\mu) \end{array}  \right)
=e^{i\pi/6}\left( \begin{array}{cc} -i[1] & [1] \\
                                         1 & 0 \end{array} \right)
\left( \begin{array}{c} g^+_{\rm right}(r,\mu) \\
g^-_{\rm right}(r,\mu) \end{array}  \right), \ \ r\in \Omega_\mu^{\rm right}. 
\]
For more on this theory, we refer the reader to the monograph \cite{fedoryuk}.  Here, like before $[1]=1+O((m\varepsilon)^{\frac12})$ and $|O((m\varepsilon)^{\frac12})|\leq C(m\varepsilon)^{\frac12}$.  The constant $C=C(r,\mu)$ depends on $r$ and $\mu$ in general.  However, given a compact $K$ in $\Omega_\mu^{\rm right}$ there exists a neighborhood $U(\mu)$ of $\mu$ such that $C=C(K)$ depends only on $K$. 

Consider again the characteristic determinant 
\[
\Delta(\mu)=\left| \begin{array}{cc} g_{\rm left}^+(r^-,\mu) & g_{\rm left}^+(r^+,\mu) \\
g_{\rm left}^-(r^-,\mu) & g_{\rm left}^-(r^+,\mu) \end{array} \right|.
\]
Using the transmission formula, we obtain 
\[
\Delta(\mu)=\frac{T(\mu)e^{i\pi/6}}{r}\left| \begin{array}{cc} [1]e^{(m\varepsilon)^{-\frac12}S(r^-,\mu)}& -i[1]e^{(m\varepsilon)^{-\frac12}(S(r^+,\mu)}+[1]e^{-(m\varepsilon)^{-\frac12}S(r^+,\mu)} \\ \ [1]e^{-(m\varepsilon)^{-\frac12}S(r^-,\mu)} & [1]e^{(m\varepsilon)^{-\frac12}(S(r^+,\mu)}\end{array} \right|.
\]
Since $\Re S(r^+,\mu)=0$ and $\Re S(r^-,\mu)<0$,  the term $e^{(m\varepsilon)^{-\frac12}S(r^-,\mu)}$ ($e^{-(m\varepsilon)^{-\frac12}S(r^-,\mu)}$) decays (grows) exponentially with $\varepsilon$, while $e^{(m\varepsilon)^{-\frac12}S(r^+,\mu)}$ remains bounded.  Therefore, $\Delta(\mu)=0$ is equivalent (up to exponentially small terms) to 
\[
e^{-(m\varepsilon)^{-\frac12}S(r^+,\mu)} -ie^{(m\varepsilon)^{-\frac12}S(r^+,\mu)} =O(\varepsilon^{\frac12}).
\]
This equation, in turn, is equivalent to 
\[
\sin\left(2(m\varepsilon)^{-\frac12}\Im Q^+(\mu)\right)=-1+O(\varepsilon^{\frac12}).
\]
If we neglect the $O(\varepsilon^{\frac12})$ term, the roots near the curve $\gamma^+$ are determined from the equation 
\[
\Im Q^+(\mu)=(m\varepsilon)^{\frac12}(k\pi-\pi/4),\ \ k\in\mathbb{Z}.
\]
The conclusion for eigenvalues near $\gamma^+$ follows similarly.  Near $\gamma^\infty$, the conclusion is simpler and doesn't require the transmission formula.  Note that the relation (\ref{Q}) for eigenvalues  bellow $(\tilde\gamma_-\backslash\gamma_-)\cup(\tilde\gamma_+\backslash\gamma_+)$ reads 
\[
\sin\left((m\varepsilon)^{-\frac12}\Im Q^\infty(\mu)\right)=O(\varepsilon^{\frac12}). 
\]
If we were allowed to neglect $O(\varepsilon^{\frac12})$, we would obtain the following formula for the solutions
\[
\Im Q^\infty(\mu)=(m\varepsilon)^{\frac12}k\pi,\ \ k\in\mathbb{Z}.
\]
Again, the details, and in particular the justification for neglecting  $O(\varepsilon^{\frac12})$ can be found in \cite{shkalikov2}.  \qed
\end{proof}
\begin{remark}
The last two theorems remain valid in the case of the disk.  The proof has to be slightly modified, however, because the approximation (\ref{fundamental_solutions}) breaks down in the neighborhood of the singularity $r=0$.  Instead, for $r\ll1$, (\ref{fundamental_solutions}) can be replaced by the following approximation 
\begin{equation}\label{approximate_solutions1}
g_{\rm{app}}^{\pm}(r,\mu)=\frac{r^{m}}{\sqrt[4]{{i}\left(\omega(r)-\mu\right)}}e^{\pm (m\varepsilon)^{-\frac12}S(r,\mu)}.  
\end{equation}
Recall that the boundary condition at $r=0$ is set to  {$\lim_{r\to0^+}  \frac{g_{m,n}(r)}{r^m}=1$}.  The discussion about the sign of the characteristic determinant is the same as in the case of the annulus.  
\end{remark}
\begin{remark}
We are particularly interested in the scaling with respect to $\varepsilon$ of $\mu_{p_\pm}$ (recall that it has the  maximal imaginary part on that branch).  Note that $\mu^\pm_{p_\pm}\to a^\pm$  (and in particular $\Im \mu^\pm_{p_\pm}\to 0$) as $\varepsilon\to0$.  Two generic examples for which the functions $Q^\pm$ can be computed explicitly are $\omega(r)=a^-+(r-r_0)$ and $\omega(r)=a^-+(r-r_0)^2$. For the former, one finds easily that $Q^\pm(\mu)=\pm\frac{2 e^{i\pi/4}}{3}(a^\pm-\mu)^\frac32$, leading to $\mu_k^\pm=a^\pm\mp e^{\pm i\pi/6}(m\varepsilon)^\frac{1}{3}r_k$, where $r_k=\left(\frac{3\pi}{2}\left(k-\frac{1}{4}\right)\right)^\frac{2}{3}$. In this case, $\mu_1^\pm-a^\pm\sim\varepsilon^\frac{1}{3}$ as $\varepsilon\to0$. In the latter case, one easily obtains $Q^-(\mu)=\frac{\pi}{4}e^{3\pi/4}(\mu-a^-)$ and $\mu_k^-=a^-+(1-4k)(m\varepsilon)^\frac{1}{2} e^{-i\pi/4}$. In this case, $\mu^-_0-a^-\sim\varepsilon^\frac{1}{2}$ as $\varepsilon\to0$. It can be verified easily that these two types of scalings for $\mu^-_{p_-}$ depend on the local behavior of $\omega$ in the neighborhood of $r_0$.  In other words,  $\omega(r)=a^-+(r-r_0)+o(|r-r|^2)$ (we will refer to this as `locally linear') leads to the scaling $\mu_{p_-}^--a^-\sim \varepsilon^\frac{1}{3}$ while $\omega(r)=a^-+(r-r_0)^2+o(|r-r|^3)$ (we refer to it as `locally quadratic') leads to the scaling $\mu_{p_-}^--a^-\sim \varepsilon^\frac{1}{2}$. 

In the case of a regularized vortical flow, $\omega(r)=\frac{1}{a^2+r^2}$ is locally  quadratic in the neighborhood of $r^-=0$ (disk) leading to $\mu_{p_-}^--a^-\sim \varepsilon^\frac{1}{2}$, and locally linear in the neighborhood of $r^->0$ (annulus) leading to $\mu_{p_-}^--a^-\sim \varepsilon^\frac{1}{3}$.  For the branch near $a^+$ it is locally linear in either case leading to $\mu_{p_+}^+-a^+\sim \varepsilon^\frac{1}{3}$.  The situation is similar in the case of the two-dimensional creeping Couette flow between two concentric cylinders of radii $0<r^-<r^+$ with the outer cylinder moving with velocity $\Omega r^+$.  The velocity field is given by 
\[
v_\theta (r)=\Omega r\frac{1-(r^-/r)^2}{1-(r^-/r^+)^2},
\]
so that the potential $\omega(r)=v_\theta(r)/r$ is increasing and it is locally linear at both $r^\pm$ leading to $\mu_{p_\pm}^\pm-a^\pm\sim \varepsilon^\frac{1}{3}$ on both branches.  
\end{remark}
\begin{remark}
The above developed theory becomes much more involved if $\omega$ is not monotonic. However, the problem is still tractable if $\omega:[r^-,r^+]\to\mathbb{R}$ is such that it decreases on $[r^-,r^c]$ and increases on $[r^c,r^+]$.  Assume for simplicity that $a^c=\omega(r^c)<a^-=\omega(r^-)<a^+=\omega(r^+)$, and let  $\Pi^\pm=\{\mu\ \big|\ \Im(\mu)<0,\ \Re(\mu)\in(a^c,a^\pm)\}$.  Let $G^\pm$ be the preimage of $\Pi^\pm$ under $\omega$ as before.  The equation $\omega(r)=\mu$ has two unique roots $r_\mu^\pm\in G^\pm$.   As before, we define the functions 
\[
Q^\pm_-(\mu)= \pm\int_{r^-_\mu}^{r^\pm}\sqrt{i\left(\omega(\xi)-\mu\right)}\ d\xi\ \  {\rm and}\ Q^\pm_+(\mu)= \pm\int_{r^+_\mu}^{r^\pm}\sqrt{i\left(\omega(\xi)-\mu\right)}\ d\xi
\]
and additional functions 
\[
Q^c(\mu)= \int_{r_\mu^-}^{r_\mu^+}\sqrt{i\left(\omega(\xi)-\mu\right)}\ d\xi\ \  {\rm and}\ Q^\infty(\mu)= \int_{r^-}^{r^+}\sqrt{i\left(\omega(\xi)-\mu\right)}\ d\xi, 
\]
and similarly as before the lines $\tilde\gamma_\pm^\pm$, $\tilde\gamma_c$ and $\tilde\gamma_\infty$ and let $\tilde\Gamma$ be the union of these six lines.  The spectral limit graph $\Gamma$ has a much more complicated structure (it is a subset of $\tilde\Gamma$), and it is beyond the scope of this paper to describe the whole structure. However, there are three lines of $\Gamma$ emerging from $a^c$, $a^-$ and $a^+$: $\gamma_c$, $\gamma_-^-$ and $\gamma_+^+$, respectively.  The scaling of eigenvalues with respect to $\varepsilon$ along these lines can be determined in a similar fashion as before from the local behavior of $\omega(r)$ in the neighborhood of $r^c$, $r^-$ and $r^+$, respectively.  

The above situation applies for example in the case of the case of the parabolic MHD-driven Poiseuille profile $v(r)=C(r-r^-)(r-r^+)$. In this case $\omega(r)= C(r-r^-)(r^+-r)/r$ and the above set-up applies with $r^c=\sqrt{r^-r^+}$.  A similar analysis as before would show that the locally quadratic behavior at $r^c$ leads to the scaling $\mu_{p_c}^c-a^c\sim \varepsilon^\frac{1}{2}$, and the locally linear behavior at $r^\pm$ leads to the scaling $\mu_{p_\pm}^\pm-a^\pm\sim \varepsilon^\frac{1}{3}$. 
\end{remark}


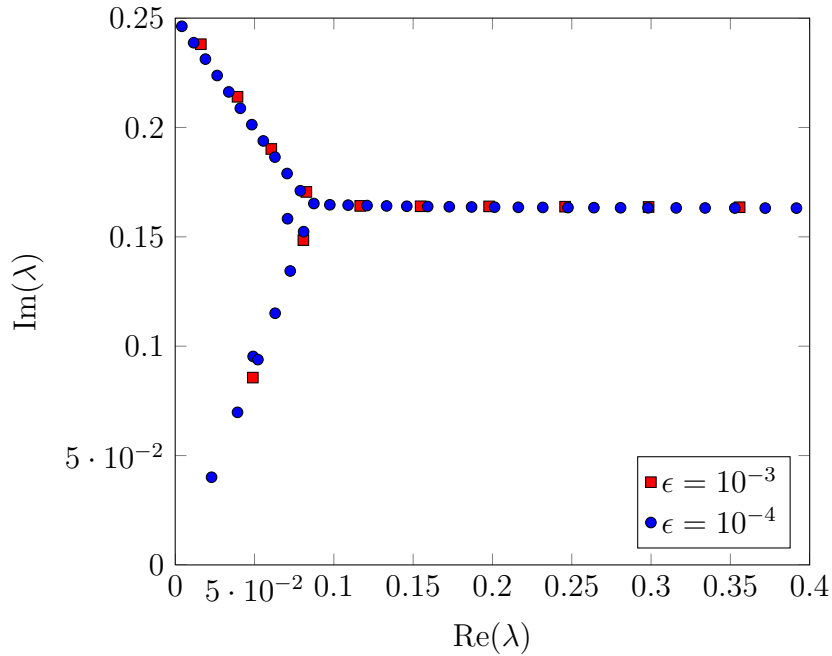
\begin{figure}[htb]
  \centering
\begin{tikzpicture}
\begin{axis}[%
scale = 1.0,
scale only axis,
xmin=0, xmax=0.4,
xlabel={$\mathrm{Re(}\lambda\mathrm{)}$},
ymin=0, ymax=0.25,
ylabel={$\mathrm{Im(}\lambda\mathrm{)}$},
legend style={at={(0.97,0.03)},anchor=south east,align=left}]
\addplot [
color=black,
only marks,
mark=square*,
mark options={solid,fill=red}
]
coordinates{
 (0.0161592953799117,0.238064396985122)(0.039289651242083,0.214031692769812)(0.0489415379309431,0.0856438755330143)(0.0605874845086697,0.190175254204576)(0.0808050728860578,0.148502362752823)(0.0826698675656788,0.170525483303019)(0.116652391563231,0.164137617833728)(0.154709618085381,0.163970637091944)(0.197890580809375,0.163918806206389)(0.245804032724008,0.163782839198718)(0.298433480807554,0.163644597705062)(0.355826768005361,0.163526393326158) 
};
\addlegendentry{$\epsilon = 10^{-3}$};
\addplot [
color=black,
only marks,
mark=*,
mark options={solid,fill=blue}
]
coordinates{
 (0.00417852863069767,0.246249648405262)(0.0116365961732929,0.238749891049587)(0.0190531613746348,0.231250788822956)(0.0229238315722611,0.0400845506859401)(0.0264286494535782,0.223752707479643)(0.0337631107622743,0.216255949180428)(0.0392835797318073,0.0697593264773816)(0.041055159534157,0.208762134155272)(0.0483038866449335,0.201281412087952)(0.0491193420359818,0.0953153001176479)(0.0521294744080335,0.0938695645832858)(0.0555340851805447,0.193842619767172)(0.0628641730131513,0.186451295872099)(0.0630087795540379,0.115067247389238)(0.0705344109136717,0.178941950821289)(0.0708113337807902,0.158279409964752)(0.0725157799852346,0.134385469358669)(0.0788969197633386,0.171063717498712)(0.0809572490909513,0.15237341486265)(0.0874111463888288,0.165219478078568)(0.0974769568099648,0.16464781633809)(0.109019519501445,0.164489114472679)(0.120938256522444,0.164295202498869)(0.133258380054352,0.164125636944967)(0.14599973418076,0.163979302250967)(0.159172638989516,0.163853103577884)(0.172784477616624,0.16374414933712)(0.186841323362028,0.163649912130082)(0.201348430969099,0.163568220734194)(0.216310417657222,0.163497225765996)(0.231731357097835,0.163435360197608)(0.247614846800813,0.163381299921384)(0.263964064881162,0.163333927106334)(0.280781820306519,0.163292297711719)(0.298070597709588,0.163255613499759)(0.315832597167691,0.163223198335288)(0.334069769309832,0.163194478320224)(0.352783846136376,0.163168965219409)(0.371976367970423,0.163146242674581)(0.39164870697179,0.16312595474284) 
};
\addlegendentry{$\epsilon = 10^{-4}$};
\end{axis}
\end{tikzpicture}%
  \caption{Spectral graph for the Poiseuille profile for two values of $\epsilon$.}
  \label{fig:spectrum}
\end{figure}

We illustrate the results by numerically computing the spectrum for the Poiseuille profile with
$r^{-} = 0.25$ and $r^+ = 1.$  using a standard Chebyshev polynomial ($N = 84$) co-location scheme
and ARPACK to solve the resulting eigenvalue problem 
(\ref{eq:eigenvalue}). 
The spectral graph (Fig. 1) for $m=1$ shows the two main solution branches localizing on the
right end point and the critical point at $r = 1/2$ as shown in the left panel of Fig.\ 2.  
The right panel of Fig.\ 2 clearly indicates the expected scaling with $\epsilon$ for the two
branches.

\begin{figure}[htb]
  \centering
  \includegraphics[width=\textwidth]{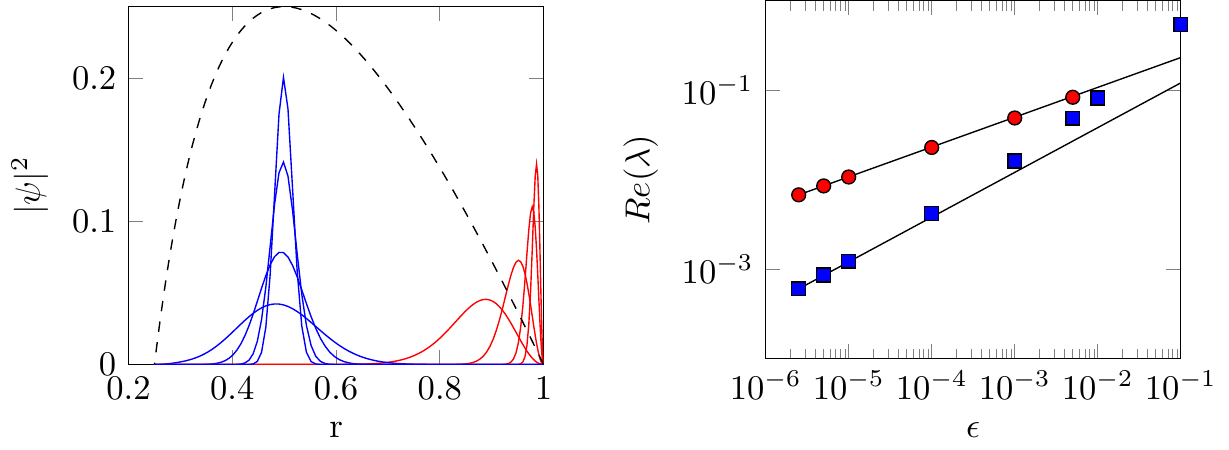}
  \caption{Left panel: Localization of eigenfunctions for the spectral branch emerging from the
  right end point (red) and the quadratic critical point (blue). Eigenfunctions shown for
$\epsilon = [10.^{-3},10^{-4},10^{-5},2.5\times 10^{-6}]$ with $w(r)$ shown in black. 
The right panel shows scaling 
of each branch along with the prediction $\sim\epsilon^{1/3}$ and $\sim\epsilon^{1/2}$.}
  \label{fig:scaling}
\end{figure}

\section{Numerical comparison of full and averaged dynamics.}
\label{sec:numerical-results}
In this section, we compare numerical solutions of 
the original equation (\ref{eq:AD_ready_for_averaging}) and the averaged equation to the first
order (\ref{eq:AD_averaged}). We consider the evolution of the
tracer field on a unit disk $(0\le r \le 1, 0\le \theta \le 2\pi)$
with zero Dirichlet boundary conditions $v(r=1,\theta) = 0$.  We
compare two solutions of the equations
(\ref{eq:AD_ready_for_averaging}) and (\ref{eq:AD_averaged}) at
Poincar\'{e} sections where $F = 0$. We use Chebychev spectral methods
to numerically approximate spatial differentiation operators and a
second order Crank-Nicolson finite difference scheme in time. For
numerics the following parameters were chosen: $a = 0.05, f(t) =
\sin(2\pi t/T)$ and $ T = 1$. For this choice of the advective force
parameters $\avg{F}$ and $\avg{F^2}$ are found to be $1/2\pi$ and
$3/8\pi^2$ correspondingly. Fig~\ref{fig:initial_10periods} represents
10 periods of evolution for some initial state, which is taken to be 
$v_0(r,\theta) = r e^{-br^2}\cos(\pi r/2)$.
\begin{figure}[h!]
  \centering
  \includegraphics[scale=.6]{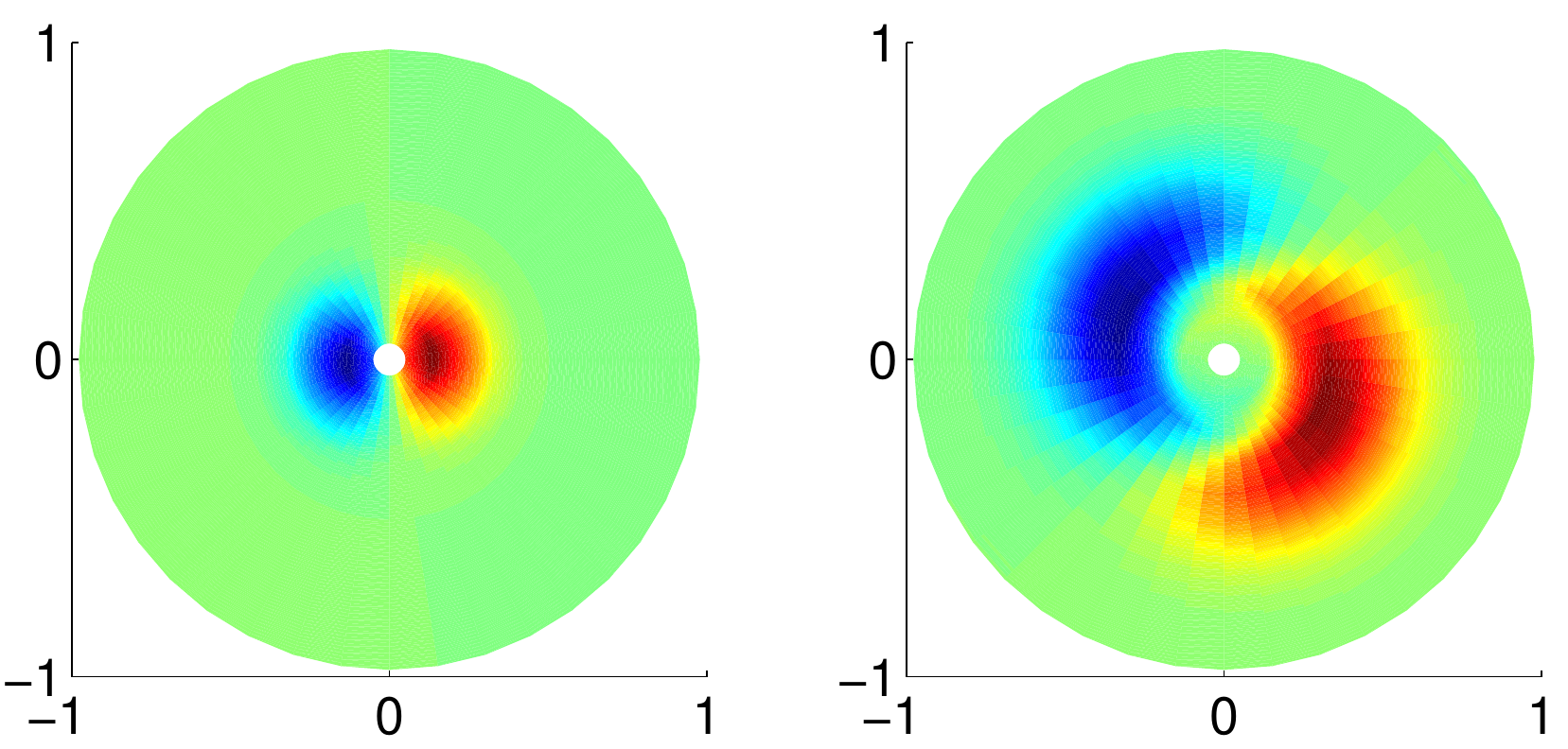}
  \caption[Evolution of the scalar field in time-dependent vortical
  field.]{Evolution of the scalar tracer field in a time-dependent
    vortical velocity field. The figure on the left represents the
    initial condition, the figure on the right shows the state of the
    system after 10 periods.}
  \label{fig:initial_10periods}
\end{figure}

In order to get a better understanding of the differences between the
behavior of the full and approximate equations we introduce the
following operators (for both the exact and averaged equations) that
maps the scalar field between two consequent Poincar\'{e} sections:
\begin{equation}
  \label{eq:1period_operator}
  \begin{array}{c} 
    \mathcal{Q}:\qquad u(r,\theta,t+T) = \mathcal{Q} u(r,\theta,t),\qquad \\ 
    \mathcal{Q}_{\mathrm{av}}:\qquad u_{\mathrm{av}}(r,\theta,t+T) = 
    \mathcal{Q}_{\mathrm{av}} u_{\mathrm{av}}(r,\theta,t) 
  \end{array}
\end{equation}
We can now study how the eigenvalues $\mathcal{Q}\psi_j =
\lambda_j\psi_j$ of the above operators change with $\varepsilon$.  To do
that we introduce relative difference in the eigenvalues of operators
as
\begin{equation}
  \label{eq:relative_difference_labmdas}
  \delta \lambda_j = \frac{|\lambda^{(\mathrm{full})}_j - 
    \lambda_j^{(\mathrm{av})}|}{\lambda_j^{(\mathrm{full})} }
\end{equation}
Here $\lambda^{(\mathrm{full})}_j$ and $\lambda^{(\mathrm{av})}_j$ are
$j$-th eigenvalues of the full and averaged operator.  In
Fig.~\ref{fig:lambdas_vs_epsilon_smallDT} we present $\delta
\lambda_j(\varepsilon)$ for several eigenmodes.

\begin{figure}[h!]
  \centering
  \includegraphics[scale=.6]{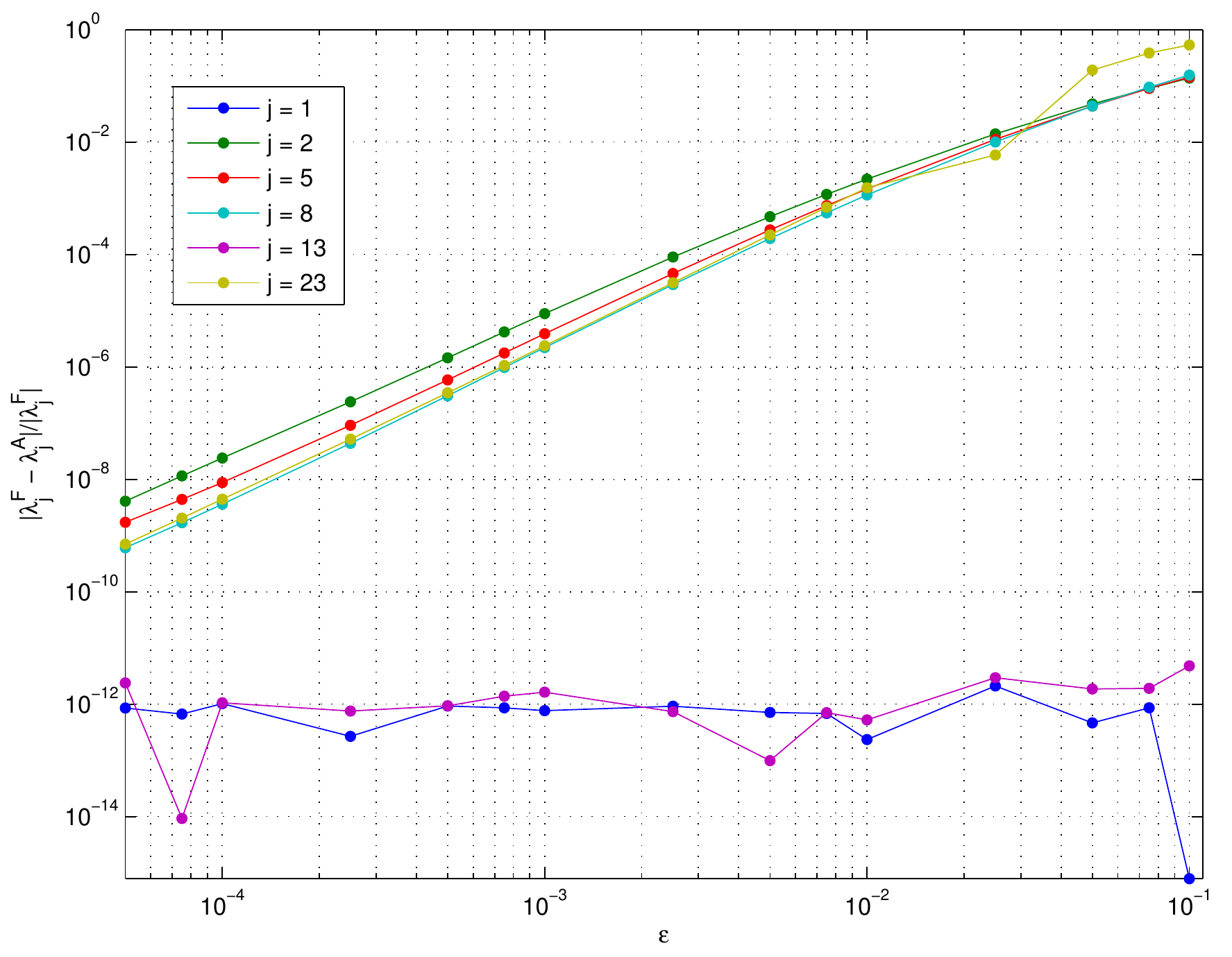}
  \caption[Difference in the calculation of eigenvalues of full and
  average operator as a function of parameter $\varepsilon$.]
  {Difference in the calculation of eigenvalues
    of the full and average operators (as given by
    (\ref{eq:relative_difference_labmdas})) as a function of the parameter
    $\varepsilon$.}. 
\label{fig:lambdas_vs_epsilon_smallDT}
\end{figure}
It is immediately seen from this plot that modes 1 and 13 are almost
identical for any value of the diffusion. This follows from the
observation that these modes possess axial symmetry, and, therefore,
the procedure of averaging does not have any effect on the one-period
evolution. For the modes that possess axial symmetry
advection-diffusion equation (\ref{eq:initial_ad_eqn}) reduces to heat
equation since advection implies only rotational translation. We also
conclude that $\delta \lambda_j \sim \varepsilon^\alpha$ where $\alpha$
is found to be equal $\alpha = 0.82$.

We now consider flow with time dependence in the form 
\begin{equation}
  \label{eq:flow_field_with_mean}
  f(t) = f_0(t) + f_1, 
\end{equation}
where $f_0(t)$ is periodic and mean-free function of time and $f_1$ is
a constant. The motion
corresponds to the rotation of the system as a whole with a constant
angular velocity $\omega$ (which still is a function of $r$) and
periodic oscillations superposed with this rotational motion. Because
of the fact that rotational motion is dependent upon $r$, large
gradients are constantly created in the scalar field. These gradients
are exposed to the action of diffusive smearing. The enhanced
stretching of the tracer field creates somewhat richer dynamics and
provides for faster mixing. We demonstrate evolution of the initial
state for the case of the flow (\ref{eq:flow_field_with_mean}) and
$\varepsilon = 0.01$ in the Fig.~\ref{fig:meanF_Evolution1}. 
\begin{figure}[htb]
  \centering
  \subfigure[$t = T$]{
    \includegraphics[width=1.1in]{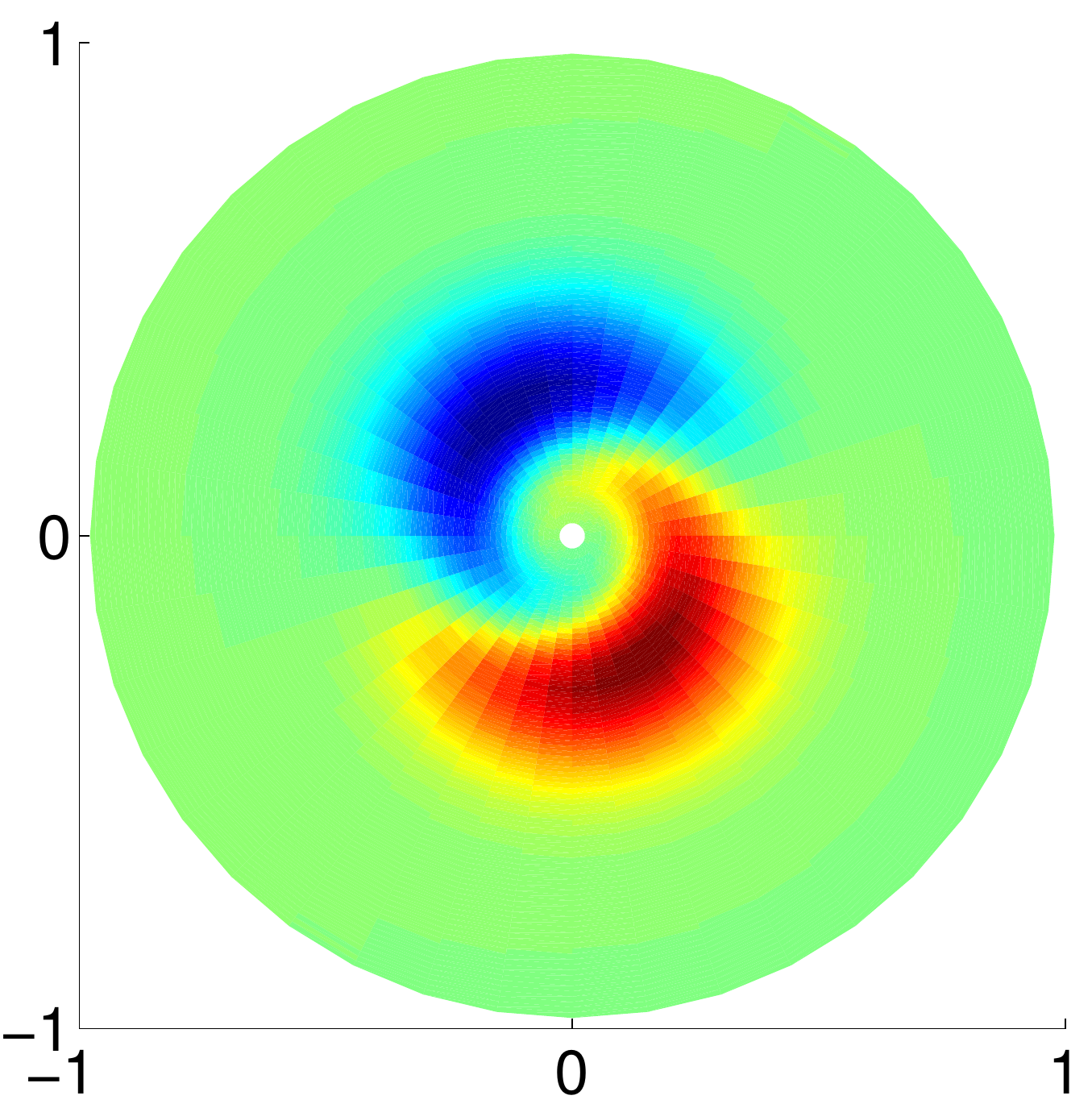}
  }
  \subfigure[$t = 10T$]{
    \includegraphics[width=1.1in]{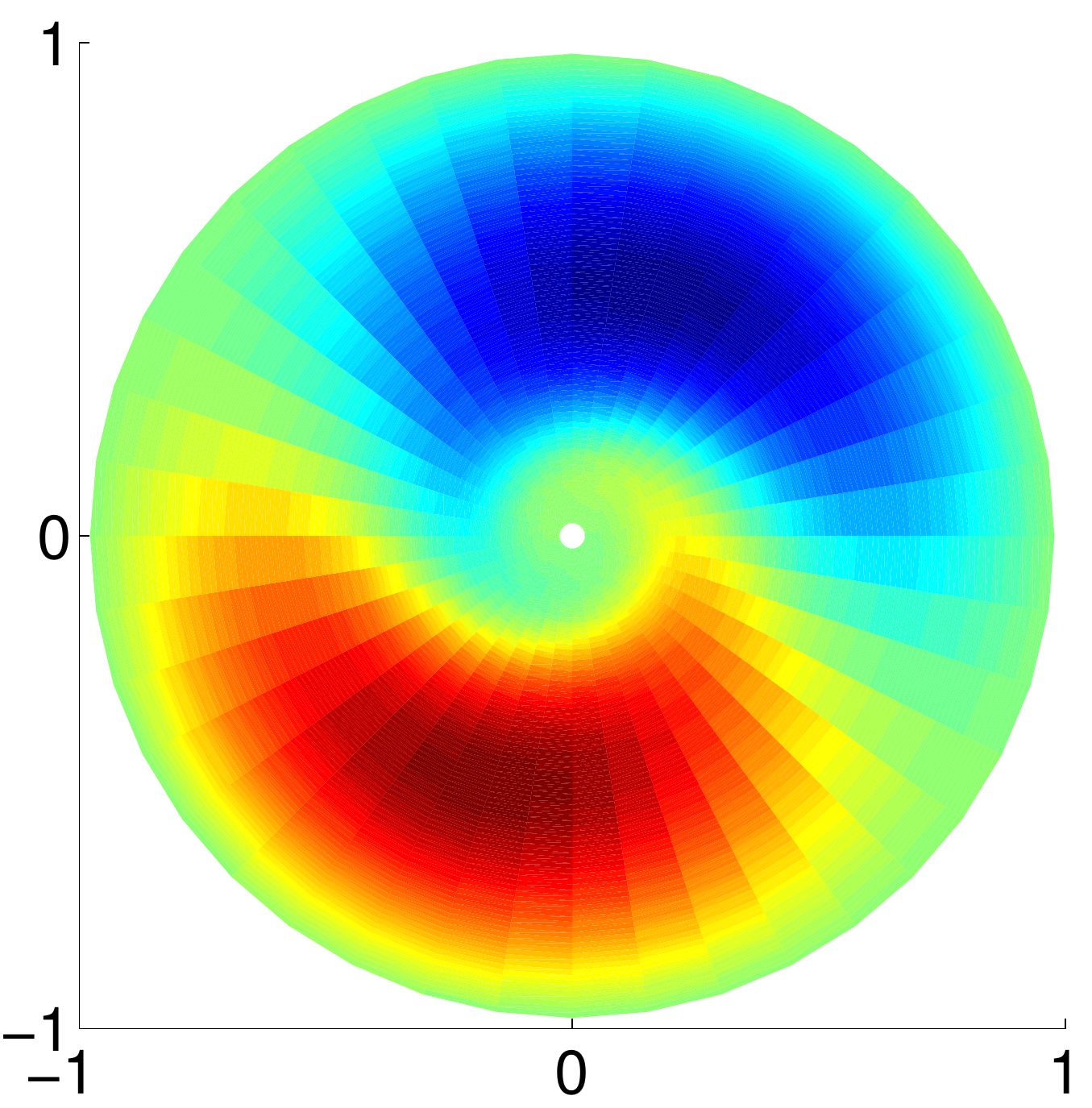}
  } 
  \subfigure[$t = 20T$]{
    \includegraphics[width=1.1in]{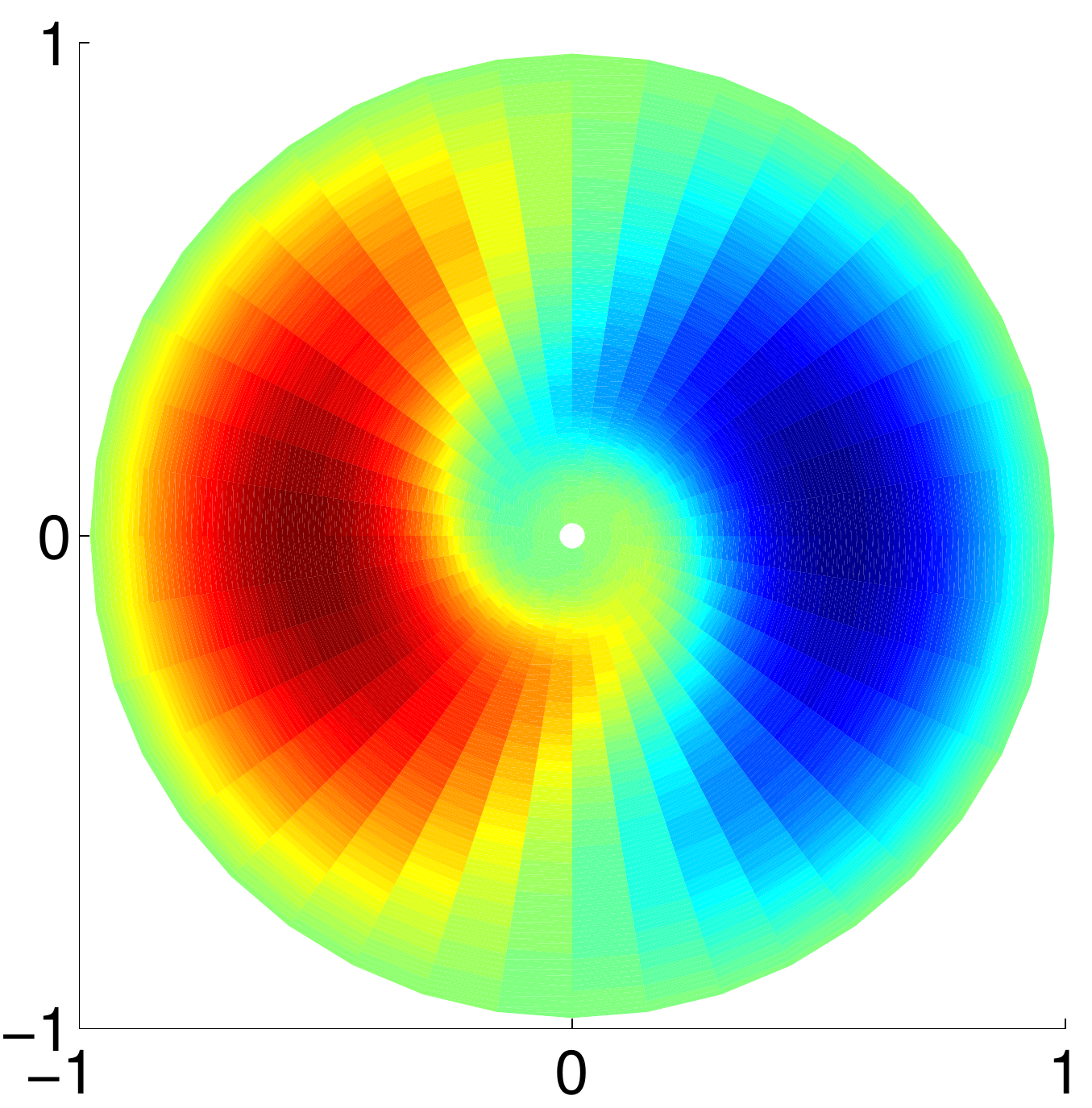}
  }
  \subfigure[$t = 30T$]{
    \includegraphics[width=1.1in]{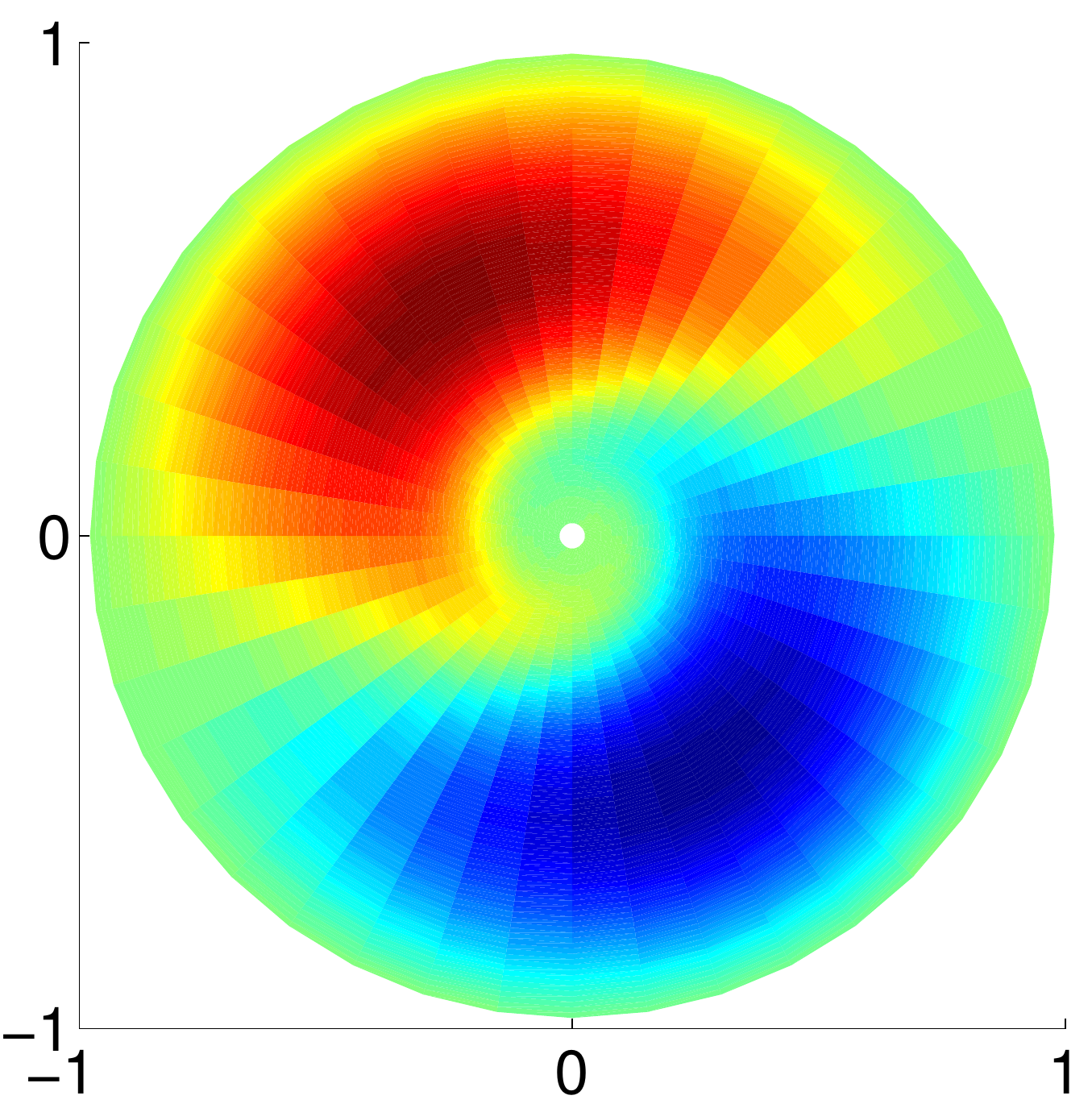}
  } \\
  \subfigure[$t = T$]{
    \includegraphics[width=1.1in]{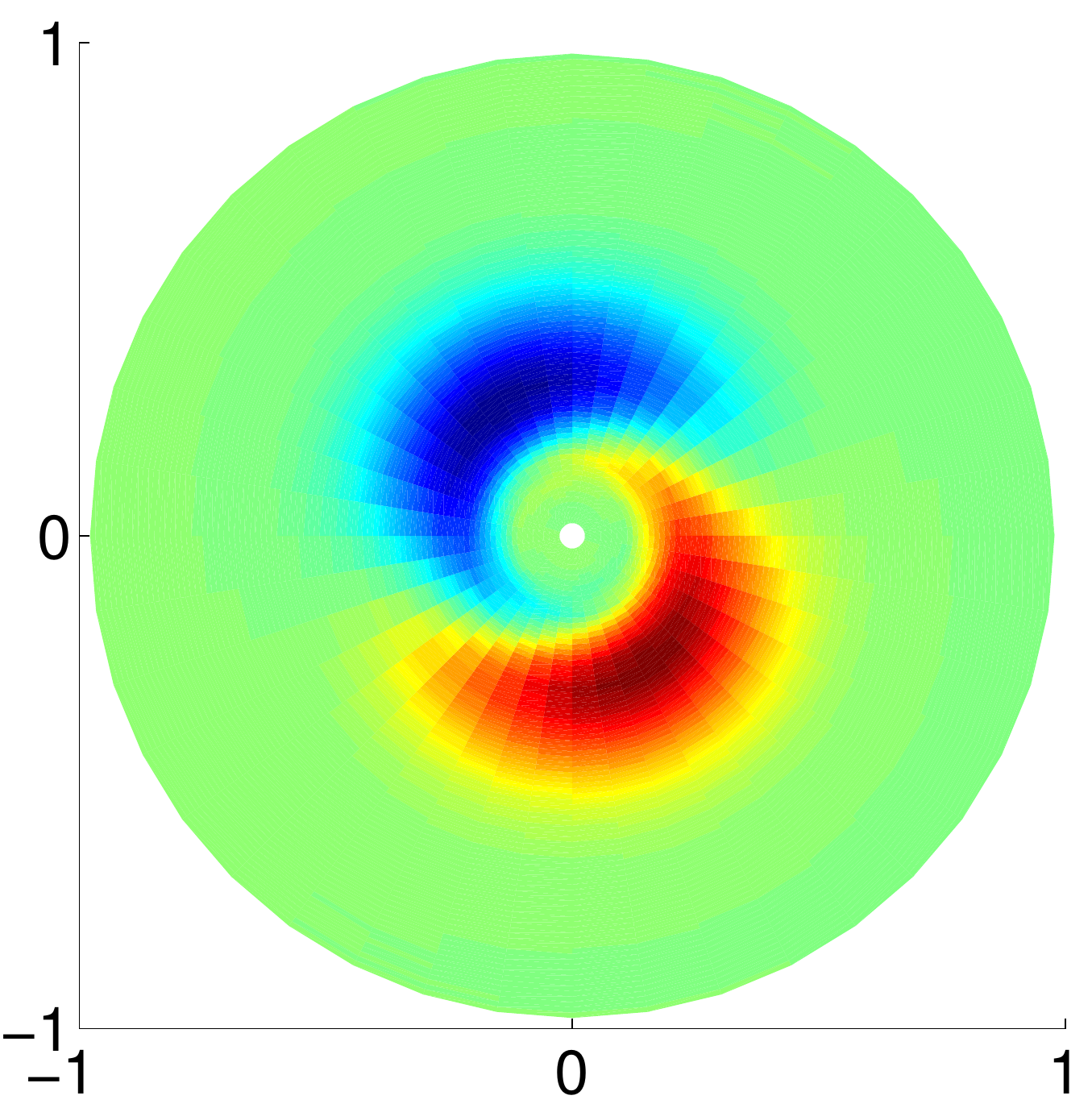}
  } 
  \subfigure[$t = 10T$]{
    \includegraphics[width=1.1in]{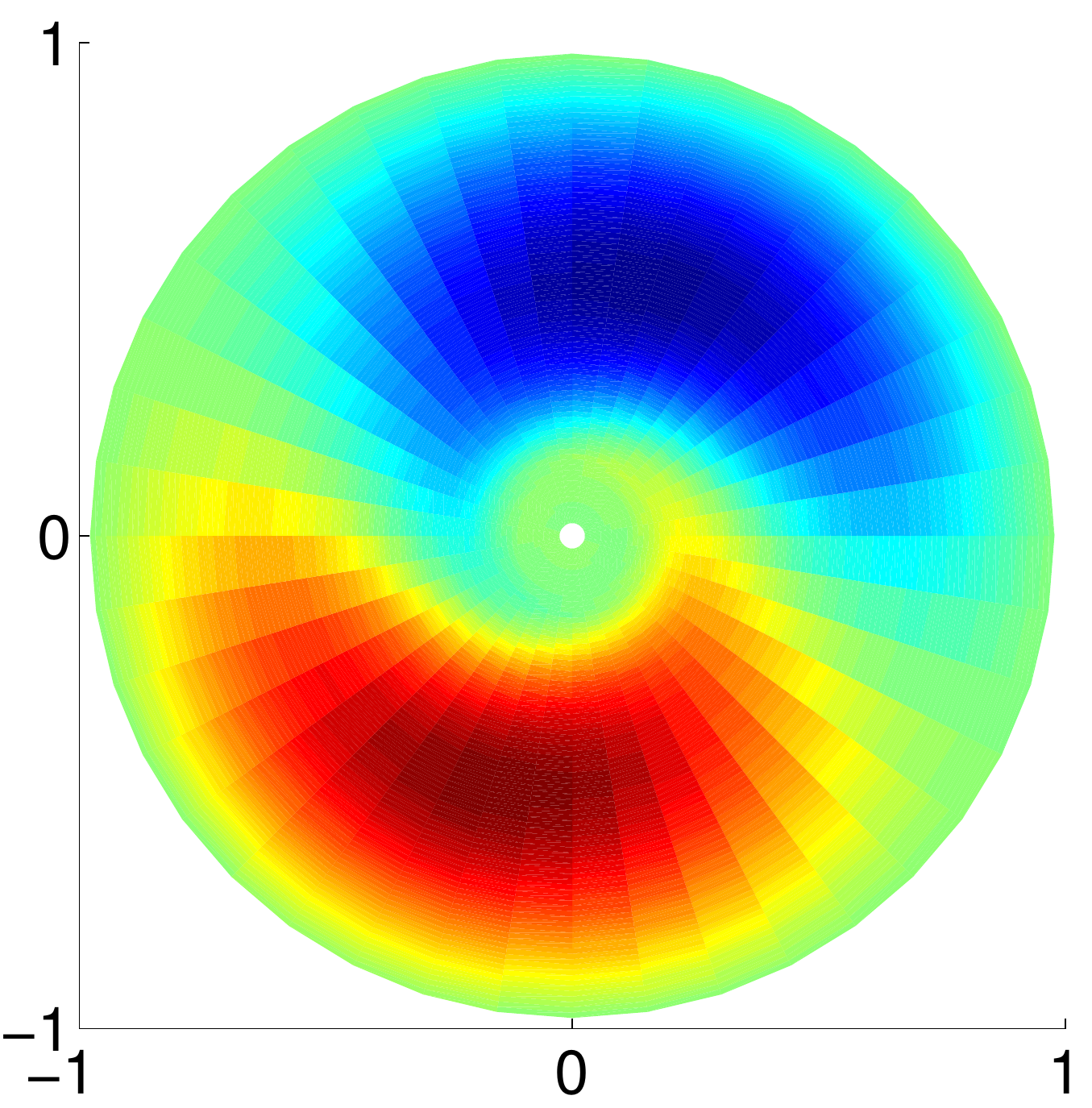}
  }
  \subfigure[$t = 20T$]{
    \includegraphics[width=1.1in]{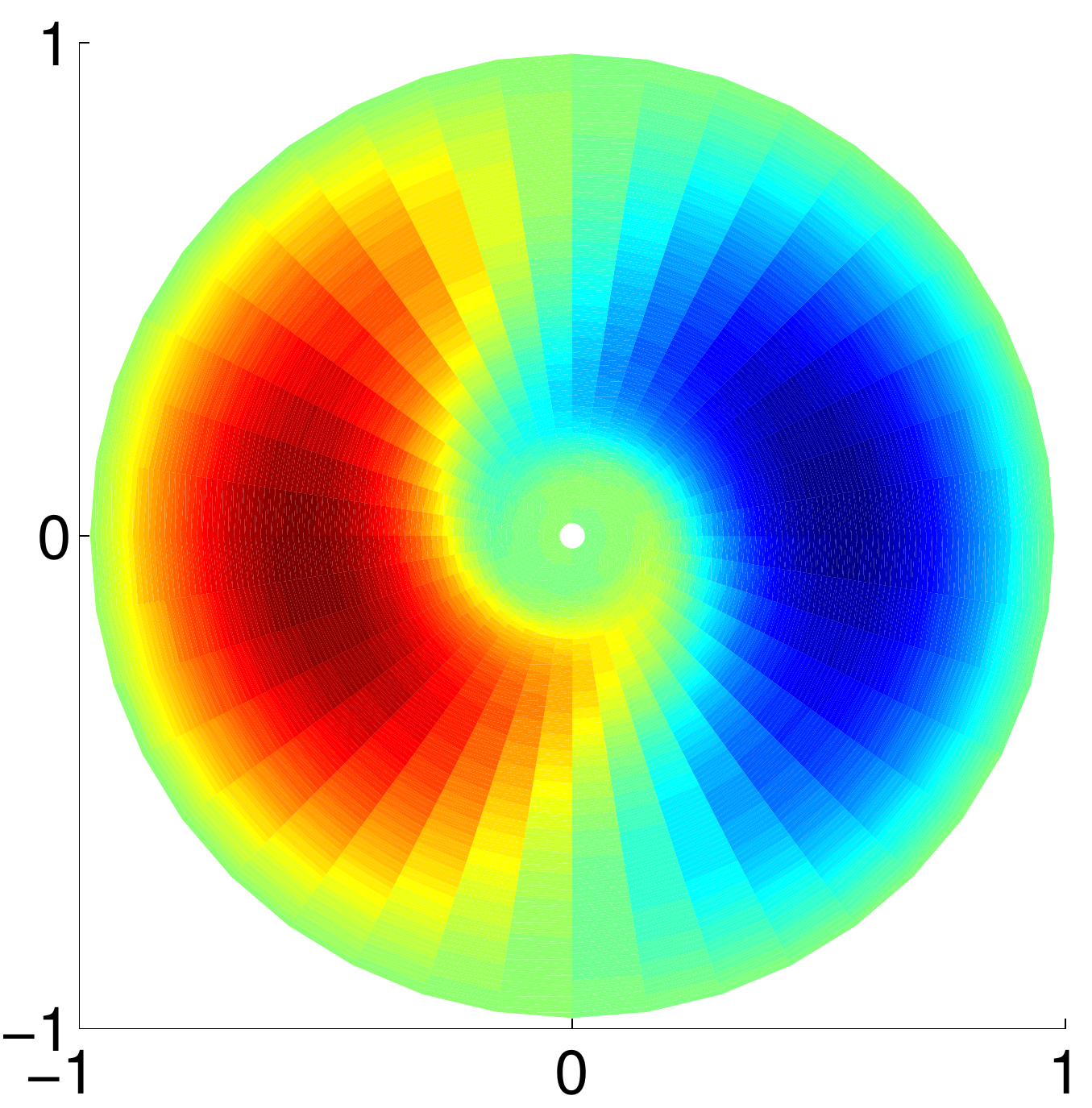}
  }
  \subfigure[$t = 30T$]{
    \includegraphics[width=1.1in]{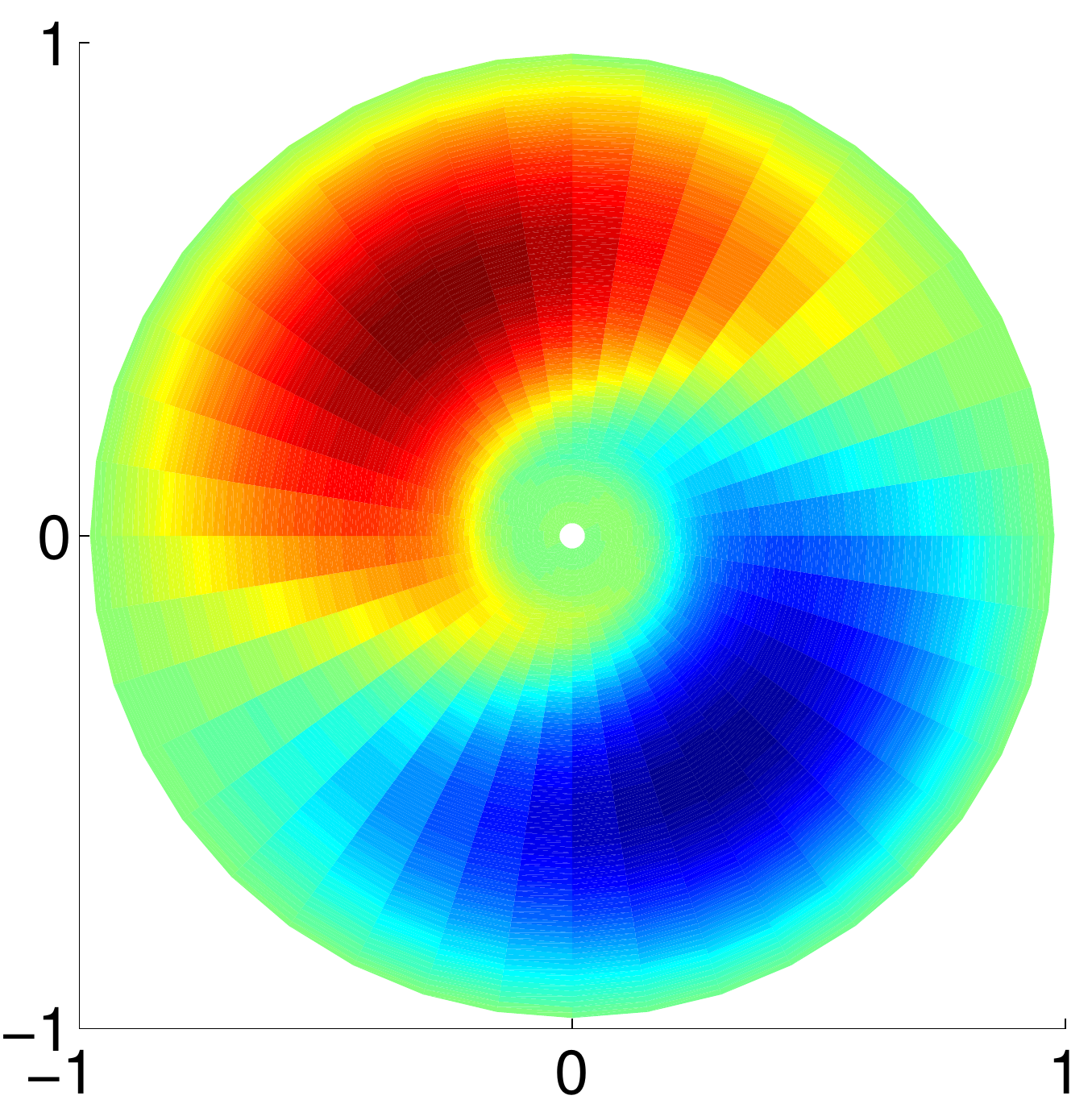}
  } 
  \caption[Full time-dependent and averaged dynamics of the tracer
  under the influence of the constant mean flow.]
  {Full time-dependent and averaged dynamics of the tracers under the
    influence of the constant mean flow shown for 1, 10, 20, and
    30 periods. The top row shows results for the full operator and bottom
    row for the averaged evolution.}
  \label{fig:meanF_Evolution1}
\end{figure}

Clearly, states as computed using full time-dependent operator and
averaged operator are almost indistinguishable. In fact, convergence
of the averaged solution to the true solution, defined by
\begin{equation}
  \label{eq:l2-norm_definition}
  ||v - v_{\mathrm{av}}|| = \left(\frac{\int_{\Omega}|v - v_{\mathrm{av}}|^2\,dx\,dy}
    {\int_{\Omega} |v|^2\,dx\,dy}\right)^{1/2},
\end{equation}
is as good as for mean-free field. For this this particular case is
$||v - v_{\mathrm{av}}|| \sim\varepsilon^{0.88}$ (see
Fig.~\ref{fig:meanF_converg})

\begin{figure}[htb]
  \centering
  \includegraphics[scale=.5]{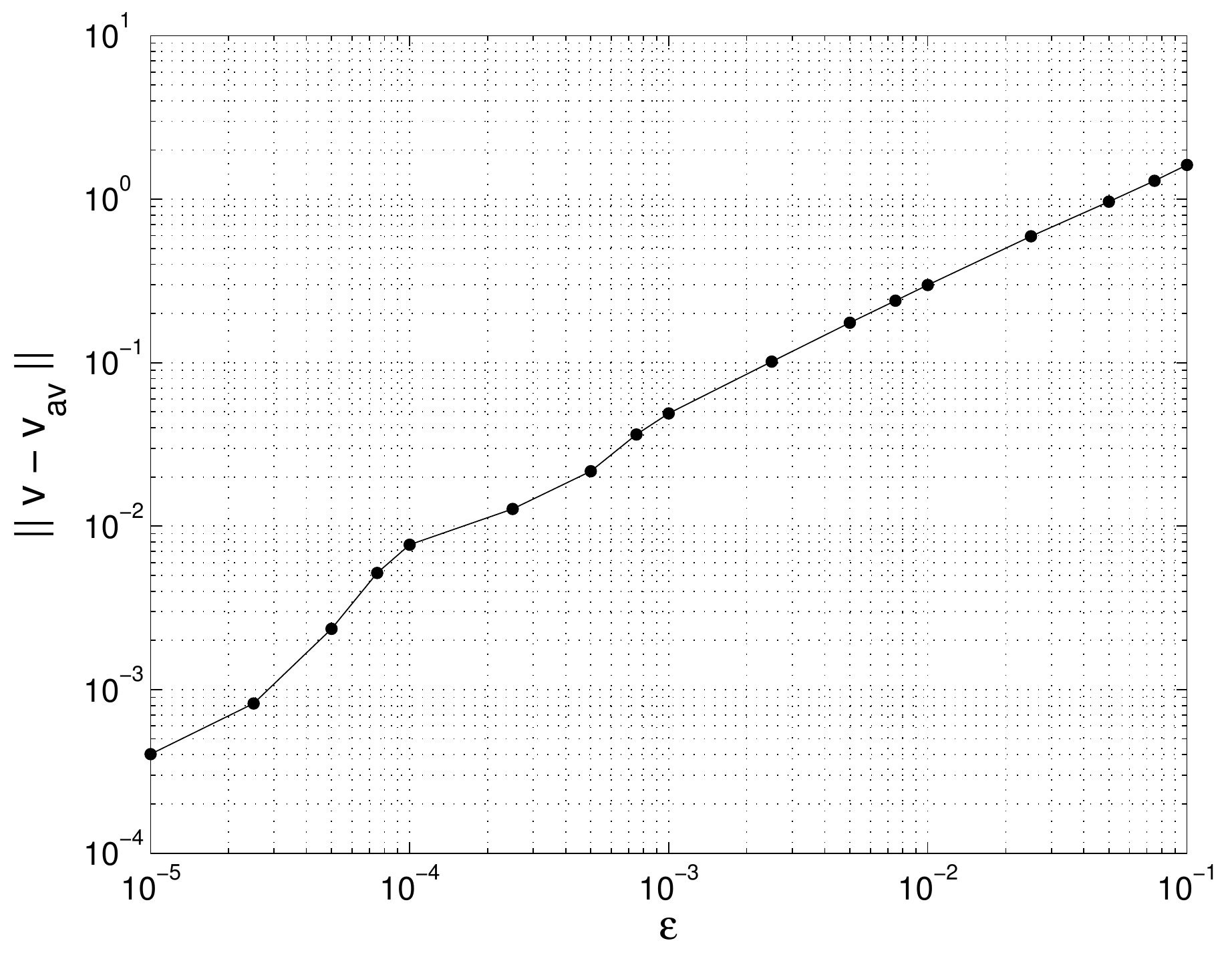}
  \caption[$L^2$ difference norm for the case of a flow with mean
  rotational component.]
  { $L^2$-norm of a difference between solutions to full and
    approximate equations as a function of parameter
    $\varepsilon$. Convergence rate is found to be
    $\sim\varepsilon^{0.88}$.}
  \label{fig:meanF_converg}
\end{figure}

We finally show evolution of the system for the case of very small
effective diffusivity $\varepsilon = 10^{-5}$. For such small value of
the parameter $\varepsilon$, tracer field does not diffuse trough the
boundary for a long time, and, therefore, large twists can be created
by the mean component of the circular flow. We illustrate it in the 
Fig.~\ref{fig:meanF_Evolution2}.

\begin{figure}[htb]
  \centering
  \subfigure[$t = 20T$]{
    \includegraphics[width=1.1in]{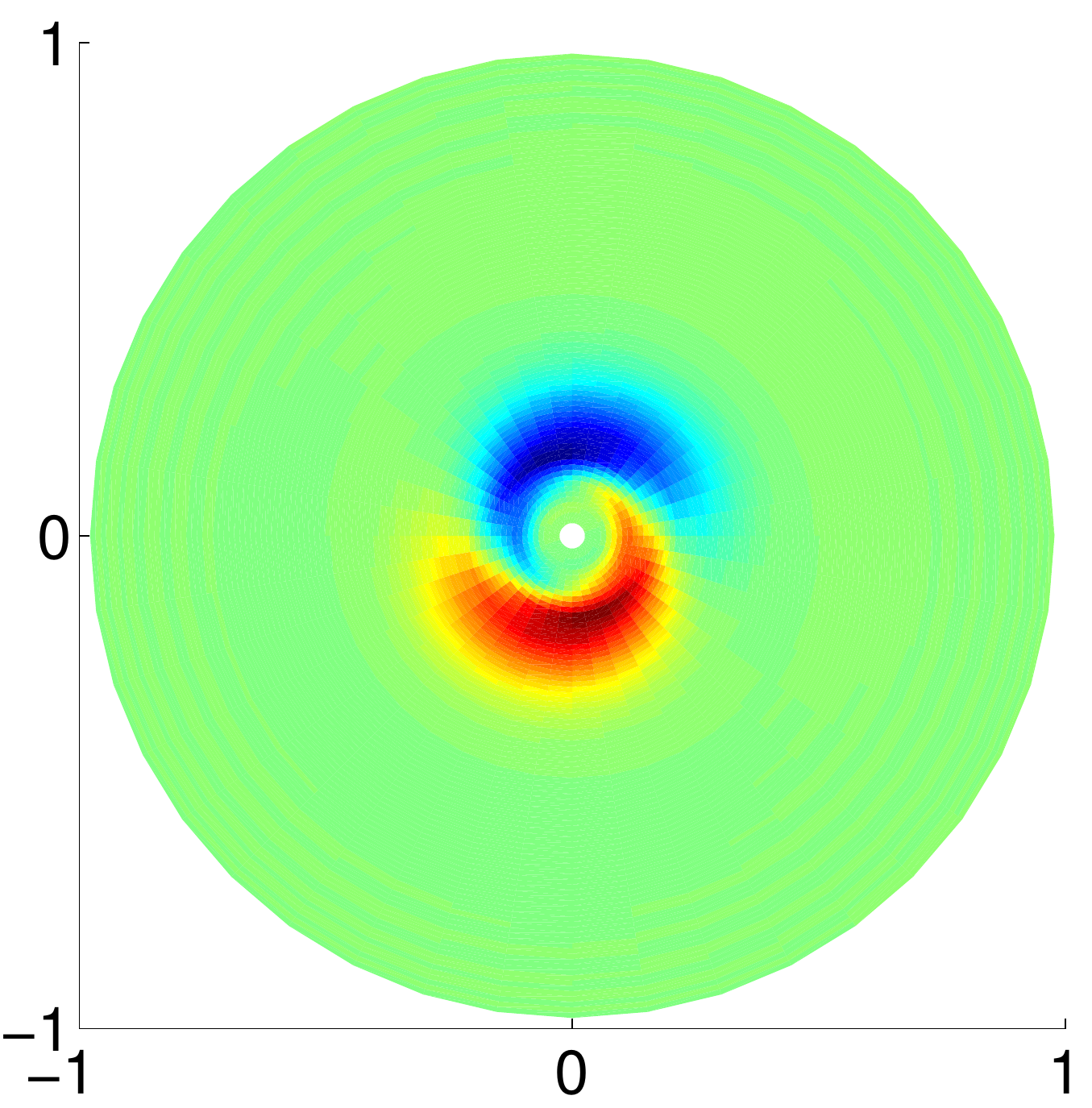}
  }
  \subfigure[$t = 100T$]{
    \includegraphics[width=1.1in]{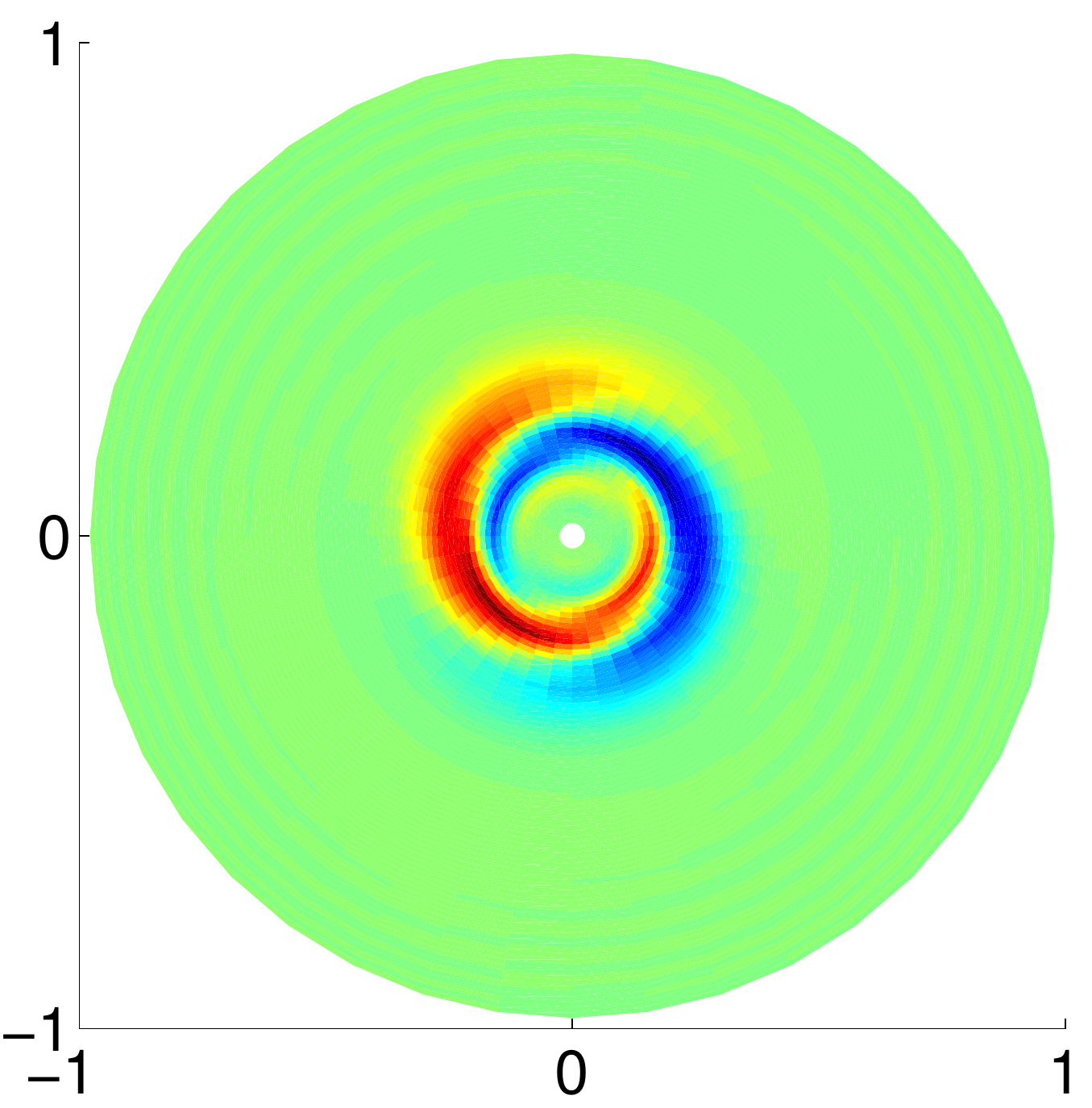}
  } 
  \subfigure[$t = 500T$]{
    \includegraphics[width=1.1in]{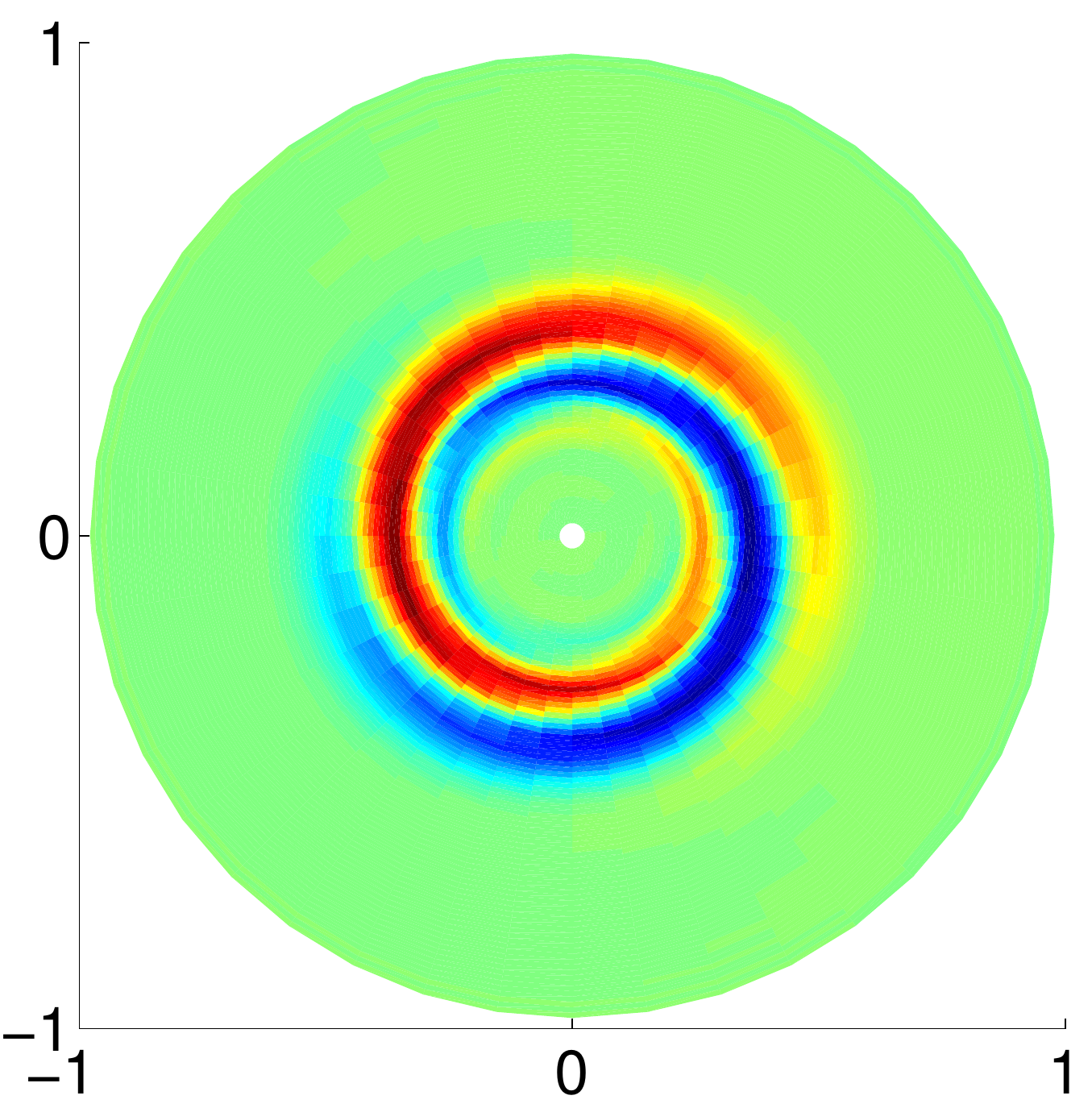}
  }
  \subfigure[$t = 1000T$]{
    \includegraphics[width=1.1in]{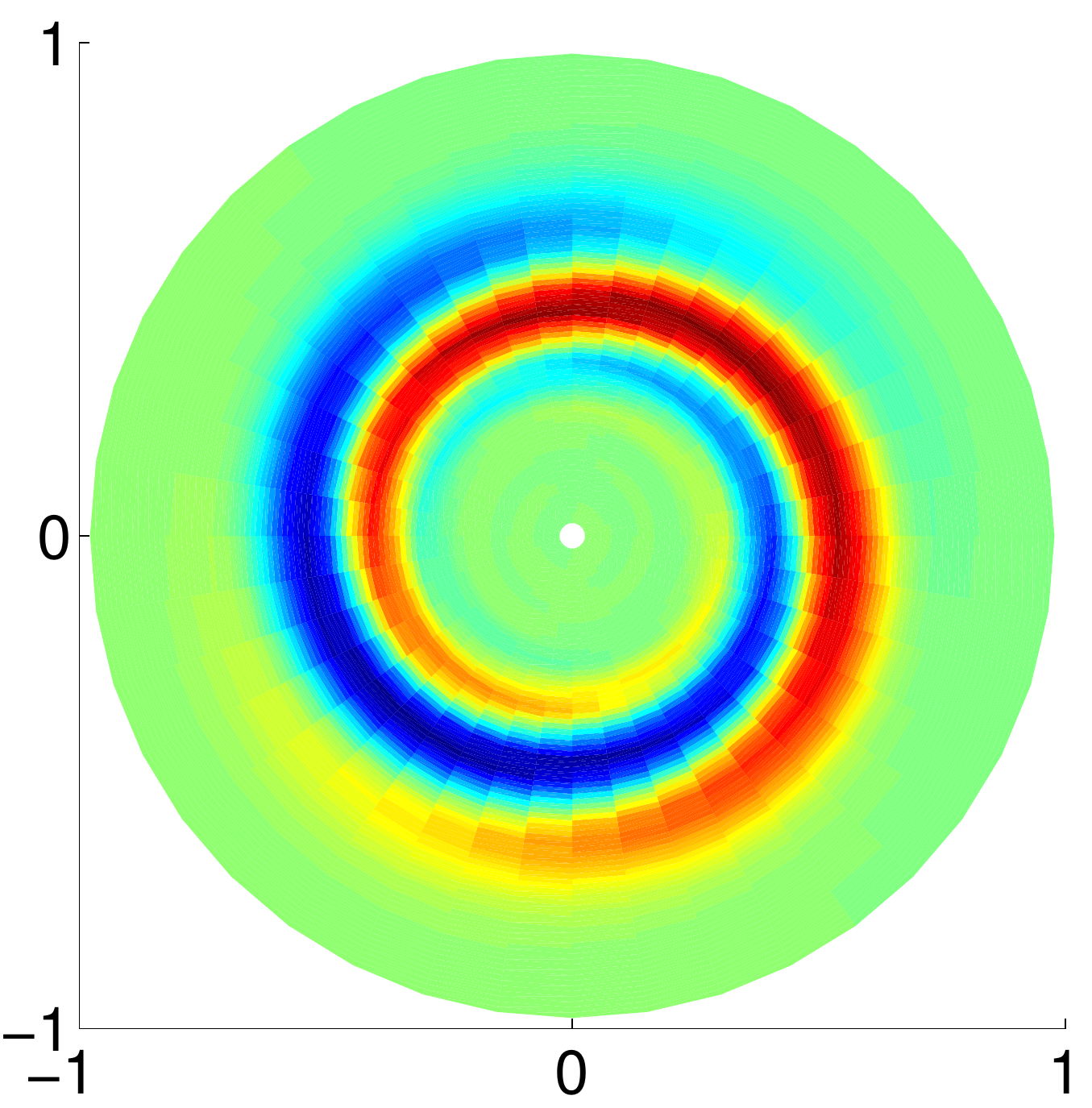}
  } 
  \caption[System evolution for small $\varepsilon$ and large times.]
  {Large-time evolution if the tracer field for the case of extremely
    small diffusivity $\varepsilon = 10^{-5}$.}
  \label{fig:meanF_Evolution2}
\end{figure}

\section{Conclusion}

In this work, we studied spectral properties of the two-dimensional advection-diffusion equation.
For a particular forms the underlying stream-function, we derive an averaged equation and, for
radial flows, we present a  characterization of the spectra of the averaged operator using complex-plane
WKBJ methods. In this way, we theoretically explain the nonlinear diffusive scaling for initial data which
are mean-free in the angle coordinate. Numerical comparison of the spectra of the full equation
and the averaged equation shows convergence of the spectra in the limit of vanishing diffusivity.

\section{Acknowledgments}

This work was partially supported by the following NSF grants: DMS-1009461, DMS-0807396, DMS-1108780, and
CNS-0855217 and the Office of Naval Research MURI OCEAN 3D + 1 grant N00014-11-1-0087.


\begin{thebibliography}{1}

\bibitem{agmon}
S. Agmon. 
\newblock {\em Lectures on Elliptic Boundary Value Problems}.
\newblock Van Nostrand, Princeton, 1965.

\bibitem{faierman}
M. Faierman. 
\newblock On the spectral theory of an elliptic boundary value problem involving an indefinite weight. 
\newblock {\em Operator Theory and Boundary Eigenvalue Problems}.
\newblock Edited by I. Gohberg and H. Langer, BirkhŠuser Verlag, Basel, 1995: 137--154.

\bibitem{batchelor}
Batchelor~G. K.
\newblock Small-scale variations of convected quantities like temperature in
  turbulent fluid.
\newblock {\em Journal of Fluid Mechanics.}, 5:113--133, 1959.

\bibitem{berestycki}
H. Berestycki, F. Hamel and N. Nadirashvili.
\newblock Elliptic eigenvalue problems with large drift and applications to nonlinear propagation phenomena.
\newblock {\em Comm. Math. Phys.}, 253:451--480, 2005.

\bibitem{const}
P. Constantin, A. Kiselev, L. Ryzhik, and A. Zlato\v s. 
\newblock Diffusion and Mixing in Fluid Flow.
\newblock {\em Annals of Math.} 168:643--674, 2008.

\bibitem{davies:2007}
 E.~B. Davies.
 \newblock {\em Linear Operators and Their Spectra}
 \newblock Cambridge University Press, New York, 2007.

\bibitem{homogenization2}
A. Fannjiang and G. Papanicolau. 
\newblock Convection enhanced diffusion for periodic flows. 
\newblock {\em SIAM Jour. Appl. Math.}, 54:333--408, 1994

\bibitem{fedoryuk}
M.~V. Fedoryuk
\newblock{\em Asymptotic analysis: Linear ordinary differential equations}.
\newblock Springer Verlag, Berlin Heidelberg, 1993

\bibitem{homogenization_thesis}
G.~A. Pavliotis.
\newblock {\em Homoenization theory for advection-diffusion equation with the
  mean flow}.
\newblock PhD thesis, Rensselaer Polytechnic Institute, Troy, New York, 2002.

\bibitem{homogenization1}
R.~M.~McLaughlin J.~Bonn.
\newblock Sensitive enhanced diffusivities for flows with fluctuating mean
  winds: A two-parameter study.
\newblock {\em Journal of Fluid Mechanics.}, 445:345 -- 375, 2001.

\bibitem{nayfeh}
A.H. Nayfeh.
\newblock {\em Perturbation methods}.
\newblock Wiley International, 1973.

\bibitem{freidlin1}
M. Freidlin and A. Wentzell.
\newblock {\em Random Perturbations of Dynamical Systems}. 
\newblock Springer-Verlag, Berlin Heidelberg, 3rd edition, 2012.

\bibitem{freidlin2}
M. Freidlin and A. Wentzell. 
\newblock {\em Random perturbations of Hamiltonian systems}.
\newblock Memoir AMS {\bf 109} 523, 1994.

\bibitem{freidlin3}
M. Freidlin and A. Wentzell. 
\newblock Diffusion Processes on Graphs and the Averaging Principle. 
\newblock {\em Ann. Prob.} 21:2215--2245, 1993.

\bibitem{freidlin4}
M. Freidlin and A. Wentzell. 
\newblock Averaging principle for quasi-linear parabolic PDEs and related diffusion processes.
\newblock{\em Stochastics and Dynamics}, 12 (01): 1150008, 20012. 

\bibitem{giona1}
M. Giona, S. Cerbelli, and V. Vitacolonna. 
\newblock Universality and imaginary potentials in advectionÐdiffusion equations in closed flows. 
\newblock {\em J. Fluid Mech.}, 513: 221--237, 2004.

\bibitem{giona2}
M. Giona, V. Vitacolonna, S. Cerbelli, and A. Adrover.
\newblock Advection diffusion in nonchaotic closed flows:
Non-Hermitian operators, universality, and localization.
\newblock {\em Phys. Rev. E}, 70:046224, 1--12, 2004. 

\bibitem{gleeson1}
J.~P. Gleeson. 
\newblock Transient micromixing: examples of laminar and chaotic stirring. 
\newblock {\em Physics of Fluids}, 17:100614, 2005.

\bibitem{gleeson2}
J.~P. Gleeson, J. West, O.~M. Roche, and A. Gelb.
\newblock Modelling annular micromixers. 
\newblock {\em SIAM J. Appl. Math.}, 64{\bf (4)}, 1294--1310, 2004.

\bibitem{koralov}
L. Koralov .
\newblock Random Perturbations of 2-dimensional Hamiltonian Flows.
{\em Probab. Theory and Related Fields}, 129: 37--62, 2004.

\bibitem{ts_lie-transform-optical}
S.~K.~Turitsyn I.~Gabitov, T.~Sch\"{a}fer.
\newblock Lie-transform averaging in nonlinear optical transmission systems
  with strong and rapid periodic dispersion variations.
\newblock {\em Phys. Lett. A}, 265:274--281, 2000.

\bibitem{lie-transform11}
Y.~Kodama.
\newblock Normal forms for weakly dispersive wave equations.
\newblock {\em Phys. Lett. A}, 112:193--196, 1985. 

\bibitem{shkalikov1}
S.~N. Tumanov and A.~A. Shkalikov.
\newblock On the limit behaviour of the spectrum of a model problem for the OrrÐSommerfeld equation with Poiseuille profile. 
\newblock {\em Izv. RAN. Ser. Mat.}, 66 {\bf(4)}:177--204, 2002

\bibitem{shkalikov2}
A.~A. Shkalikov.
\newblock Spectral Portraits of the OrrÐSommerfeld Operator with Large Reynolds Numbers. 
\newblock {\em J. Math. Sci.}, 124 {\bf(6)}: 5417--5441, 2004.



\bibitem{shkalikov3} 
V.~I. Pokotilo and A.~A. Shkalikov.
\newblock Semiclassical Approximation for a Nonself-Adjoint SturmÐLiouville Problem with a Parabolic Potential.
\newblock {\em Mat. Zametki}, 86, {\bf (3)}: 469--473, 2009.

\bibitem{ts_averaging}
Tobias Sch\"{a}fer, Andrew~C. Poje, and Jesenko Vukadinovic.
\newblock Averaged dynamics of time-periodic advection diffusion equations in
  the limit of small diffusivity.
\newblock {\em Physica D: Nonlinear Phenomena}, 238:233--240, 2009.

\bibitem{strouhal}
Frank~M. White.
\newblock {\em Fluid Mechanics}.
\newblock McGraw Hill, 4th edition, 1999.

\bibitem{arnold}
V.~I. Arnold.
\newblock {\em Mathematical Methods of Classical Mechanics}.
\newblock Springer, New York, 1989.

\bibitem{zlatos}
A. Zlato\v s.
\newblock Diffusion in fluid flow: Dissipation enhancement by flows in 2D.
\newblock{\em Comm. Partial Differential Equations}, 35:496--534, 2010.

\end{thebibliography}
\end{document}